\documentclass[10pt,journal,cspaper,compsoc]{IEEEtran}
\usepackage{times,amsmath,amssymb,graphicx,epsf,psfrag,epsfig,cite,color,multirow}
\usepackage[hyphens]{url}
\usepackage{graphicx,mathrsfs}
\usepackage{algpseudocode}
\usepackage{algorithm}
\newtheorem{theorem}{Theorem}
\newtheorem{lemma}{Lemma}

\newtheorem{proposition}{Proposition}
\newtheorem{property}{Property}
\newcommand{\eg}{\emph{e.g., }}
\newcommand{\mdef}{\overset{\Delta}{=}}
\newcommand{\expectation}{\mathbb{E}}
\newcommand{\ie}{\emph{i.e., }}

\newcommand{\mj}{\mathcal{J}}

\newcommand{\mc}{\mathcal{C}}

\newcommand{\mm}{\mathcal{M}}
\newcommand{\mb}{\mathcal{B}}

\newcommand{\mk}{\mathcal{K}}

\usepackage{amsfonts}

\IEEEaftertitletext{\vspace{-1\baselineskip}}

\ifodd 0
\newcommand{\rev}[1]{{\color{blue}#1}} %revise of the text
\newcommand{\com}[1]{\textbf{\color{red} (COMMENT: #1) }} %comment of the text
\newcommand{\comg}[1]{\textbf{\color{green} (COMMENT: #1)}}
\else
\newcommand{\rev}[1]{#1}

\newcommand{\com}[1]{}
\newcommand{\comg}[1]{}
\newcommand{\response}[1]{}
\fi
\begin{document}

\title{Dynamic Profit Maximization of Cognitive Mobile Virtual Network Operator}
%\author{Shuqin Li, Jianwei Huang, and Shuo-Yen Robert Li \thanks{The authors are with Department of Information Engineering, The  Chinese University of Hong Kong. E-mail: \{lsq007, jwhuang, bobli\}@ie.cuhk.edu.hk. Part of the results have appeared in \emph{IEEE ICC} 2012\cite{li2012profit}. This work is supported by the General Research Funds (Project Number CUHK 412710 and CUHK 412511) established under the University Grant Committee of the Hong Kong Special Administrative Region, China. Jianwei Huang is the corresponding author.} }
\author{Shuqin Li,~\IEEEmembership{Member,~IEEE,}
        Jianwei Huang,~\IEEEmembership{Senior Member,~IEEE,}
        and~Shuo-Yen Robert Li,~\IEEEmembership{Fellow,~IEEE}% <-this % stops a space
\IEEEcompsocitemizethanks{\IEEEcompsocthanksitem Shuqin Li is with Research and Innovation, Alcatel-Lucent Shanghai Bell Co., Ltd., D400, Bldg. 3, 388 Ningqiao Rd, Shanghai, 201206, China.
E-mail: Shuqin.Li@alcatel-sbell.com.cn This work was done when Shuqin Li was in The Chinese University of Hong Kong.
%\protect\\

\IEEEcompsocthanksitem Jianwei Huang is with Department of Information Engineering, The  Chinese University of Hong Kong. E-mail: jwhuang@ie.cuhk.edu.hk.
Jianwei Huang is the corresponding author.

\IEEEcompsocthanksitem Shuo-Yen Robert Li is with  both Department of Information Engineering and Institute of Network Coding, The  Chinese University of Hong Kong. E-mail: bobli@ie.cuhk.edu.hk.}% <-this % stops a space
\thanks{Part of the results have appeared in \emph{IEEE ICC} 2012\cite{li2012profit}. This work is supported by the General Research Funds (Project Number CUHK 412710 and CUHK 412511)  and AoE established under the University Grant Committee of the Hong Kong Special Administrative Region, China, and grants from China 973 Prog. No.2012CB315901 \& 2012CB315904.}}

%\begin{abstract}%\com{I have not changed the abstract.}
%We study the profit maximization problem of a cognitive virtual network operator  in a dynamic network environment. We consider a downlink OFDM communication system with various network dynamics, including dynamic user demands, uncertain sensing spectrum resources, dynamic spectrum prices, and time-varying channel conditions. In addition, heterogenous users and imperfect sensing technology are incorporated to make the network model more realistic.  By exploring the special structural of the problem,  we develop a low-complexity on-line control policies that determine pricing and resource scheduling without knowing the statistics of dynamic network parameters. We show that the proposed algorithms can achieve arbitrarily close to the optimal profit with a proper trade-off with the queuing delay.
%\end{abstract}
%
%\begin{IEEEkeywords}
%Cognitive radio networks, Dynamic network algorithm design, Network Pricing and profit maximization.
%\end{IEEEkeywords}

\IEEEcompsoctitleabstractindextext{%
\begin{abstract}
We study the profit maximization problem of a cognitive virtual network operator  in a dynamic network environment. We consider a downlink OFDM communication system with various network dynamics, including dynamic user demands, uncertain sensing spectrum resources, dynamic spectrum prices, and time-varying channel conditions. In addition, heterogenous users and imperfect sensing technology are incorporated to make the network model more realistic.  By exploring the special structural of the problem,  we develop a low-complexity on-line control policies that determine pricing and resource scheduling without knowing the statistics of dynamic network parameters. We show that the proposed algorithms can achieve arbitrarily close to the optimal profit with a proper trade-off with the queuing delay.
\end{abstract}
% IEEEtran.cls defaults to using nonbold math in the Abstract.
% This preserves the distinction between vectors and scalars. However,
% if the journal you are submitting to favors bold math in the abstract,
% then you can use LaTeX's standard command \boldmath at the very start
% of the abstract to achieve this. Many IEEE journals frown on math
% in the abstract anyway. In particular, the Computer Society does
% not want either math or citations to appear in the abstract.

% Note that keywords are not normally used for peerreview papers.
}

\maketitle

\section{Introduction}
\label{sec:intro}
The limited wireless spectrum is becoming a bottleneck for meeting today's fast growing demands for wireless data services. More specifically, there is very little spectrum left that can be licensed to new wireless services and applications. However, extensive field measurements\cite{spectrum_measurement} showed that much of the licensed spectrum remains idle most of the time, even in densely populated metropolitan areas such as New York City and Chicago. A potential way to solve this dilemma is to manage and utilize the licensed spectrum resource in a more efficient way.

This is why the concept of \emph{Dynamic Spectrum Access} (DSA) has received enthusiastic support from governments and industries worldwide\cite{FCC,akyildiz2006next,zhao2007survey}. We can roughly classify various DSA approaches into two main categories:
the \emph{spectrum sensing} based ones and the \emph{spectrum leasing} (or market) based ones. The first category indicates a hierarchical access model, where unlicensed  secondary users  opportunistically access the under-utilized part of the licensed spectrum, with controlled interference to the licensed primary users. During this process, spectrum sensing helps the secondary users to detect the currently available  spectrum resource. In contrast, the second category relates to a dynamic exclusive use model, which allows licensees to trade spectrum usage right to the secondary users. In both categories, it is possible to have a secondary operator coordinating the transmissions of multiple secondary users.

There are pros and cons for both DSA categories. Spectrum sensing detects and identifies the available unused licensed spectrum
through  technologies such as beacons, geolocation system, and cognitive radio.
Form the secondary operator's perspective, the spectrum acquired by sensing is \emph{an unreliable resource}, since it cannot determine how much resource is available before sensing. Furthermore, imperfect sensing may lead to collisions with primary users, and thus reduce the incentives for the licensee to share the spectrum.
Therefore the secondary operator needs to carefully design sensing and access algorithm to \emph{control the collision probability under an acceptable level}.
In dynamic spectrum leasing,
a secondary operator acquires the exclusive right to use spectrum within a limited time period by paying the corresponding leasing price. Thus the spectrum acquired by spectrum leasing is a reliable resource. However, the cost can be high compared to the spectrum sensing cost, and is dynamically changing according to the demand and supply relationship in the market.

In this paper, we will consider a hybrid model, where a secondary operator obtains resources from the primary licensees through both \emph{spectrum sensing} and \emph{dynamic spectrum leasing}, and provides services to the secondary unlicensed users.
Our study is motivated by \cite{duan2010cognitive,duan2011cognitive}, in which the authors introduced the new concept of Cognitive Mobile Virtual Network Operator (C-MVNO). The C-MVNO is a generalization of the existing business model of  MVNO \cite{MVNO_intro_2}, which refers to the network operator who does not own a licensed frequency spectrum or even wireless infrastructure, but resells wireless services under its own brand name. The MVNO business model has been very successful after more than 10 years' development, and there are more than 600 MVNOs today\cite{MVNO_intro,MVNO_list}. The C-MVNO model generalizes the MVNO model with DSA technologies, which allow the virtual operator to obtain spectrum resources through both spectrum sensing and leasing.
The C-MVNO model can be applied to a wild range of wireless scenarios.
One example is the IEEE 802.22 standard\cite{stevenson2009ieee}, which suggests that the cognitive radio network using white space in TV spectrum will operate on a point to multipoint basis (\ie a base station to customer-premises equipments). Such a secondary base station can be operated by a  C-MVNO.

The key difference between our work and the ones in
\cite{duan2010cognitive,duan2011cognitive} is that we study a much more realistic dynamic network in this paper. In \cite{duan2010cognitive,duan2011cognitive}, the authors formulated the problem based on a static network scenario, and provided interesting equilibrium results through a one-shot Stackelberg game. However, the real network is highly dynamic.
For example, users arrive and leave the systems randomly, the statistics of spectrum availability changes over time, and the spectrum-sensing results are imperfect. Also the leasing price is often unpredictable and changing from time to time. These dynamics and realistic concerns make the network model and the corresponding analysis rather challenging.

In this paper, we focus on the profit maximization problem for C-MVNO in a dynamic network scenario. Our key results and contributions are summarized as follows.

\begin{itemize}
    \item \emph{A dynamic network decision model:} Our model incorporates various key dynamic aspects of a cognitive radio network and the dynamic decision process of a C-MVNO. We model sensing channel availability, leasing market price, and channel conditions as exogenous stochastic processes.
        %We allow users to dynamically join the network with random demands (file sizes).
%     of a cognitive radio network, including end users' demands, availabilities of the sensing spectrum, the market price in the spectrum leasing market, and the channel conditions. %It is more comprehensive and accurate than existing static models.

\item \emph{Dynamic user demands:} We allow users to dynamically join the network with random demands (file sizes).
%    Meanwhile, the useful market information (e.g., users' willingness to pay) that can affect the demand is modeled as the time-varying rate of users' incoming process. %The operator can utilize this information to make proper dynamic pricing decisions to shape the demands for more profits.
%This model characterizes the dynamic property of practical user demands.
The demand is affected by both the transmission prices (decision variables) and market states (exogenous stochastics).
%Moreover, it properly incorporates the market information effect on pricing and demand shaping. %Further, in this model we do not need any assumption (e.g., concave) on the shape of the demand function.

\item \emph{Realistic cognitive radio model:} We incorporate various practical issues such as imperfect spectrum sensing, primary users' collision tolerance, and sensing technology selection.
%The sensing result in practice can not avoid errors, which can lead to collision with primary users.
The operator needs to choose a sensing technology to trade-off between cost and performance.
%For the operator, to choose a better sensing technique can decreases the collision odds, but it will increase the cost.  How to achieve a proper trade-off is an important part in the revenue management.

	 \item \emph{A low-complexity on-line control policy:} {By exploiting the special structure of the problem,} we design a low-complexity on-line pricing and resource allocation policy, which can achieve arbitrarily close to the operator's optimal profit. The policy does not require precise information of the dynamic network parameters, has a low system overhead, and is easy to implement.

\end{itemize}

%\com{Need to highlight the contributions more. Why a more realistic model is more difficult to analyze? If Neely is the reviewer, how do you convince him about the contribution of this paper?}

The remainder of the paper is organized as follows. In Section II, we introduce the related work. In Section~\ref{sec:system_model}, we introduce the system model. Section~\ref{sec:problem_formulation} describes the problem formulation. In Section~\ref{sec:PMC_policy}, we propose the profit maximization control (PMC) policy for homogeneous users and analyze its performance. We further extend profit maximization control policy (M-PMC policy) to heterogeneous users in Section~\ref{sec:MPMC_policy}. Section~\ref{sec:simulation} provides simulation results for both PMC and M-PMC polices. Finally, we conclude the paper in Section~\ref{sec:conclusion}.

\section{Related Work}
\label{sec:literature}

Among the vast literature on cognitive radio, we will focus on the results on \emph{operator-oriented} cognitive radio networks, where secondary operators play key roles in terms of coordinating the transmissions of the secondary users.
These studies only started to emerge recently, e.g.,\cite{al2008stackelberg,duan2010cognitive,duan2011cognitive,yu2010pricing,jia2008competitions,duan2011competition,ileri2005demand,elias2010joint,niyato2008dynamics,sengupta2007economic,gao2010map,jia2009revenue}.  We can further classify these studies into two clusters: monopoly models with one operator, and oligopoly models with multiple operators.

References \cite{duan2010cognitive,duan2011cognitive,al2008stackelberg,yu2010pricing} studied monopoly models using the Stackelberg game formulation.
%\footnote{There is another category of literature about user-oriented monopoly pricing, \eg \cite{zhang2009stackelberg,ren2011pricing,Duan:2011}.}.
%Besides \cite{duan2010cognitive,duan2011cognitive} discussed in Section~\ref{sec:intro}, both \cite{al2008stackelberg} and \cite{yu2010pricing} study the monopoly pricing based on a Stackelberg game model.
Daoud et al.~ in \cite{al2008stackelberg} proposed a profit-maximizing pricing strategy for uplink power control problem in wide-band cognitive radio networks. %, in order to maximize the profit for the service provider and ensures incentive compatibility for the users.
Yu et al.~ in \cite{yu2010pricing}
%based on the equilibrium of the Stackelberg game,
proposed a pricing scheme that can guarantee a fair and efficient power allocation among the secondary users.
%The studies in \cite{duan2010cognitive,duan2011cognitive,al2008stackelberg,yu2010pricing} focused on a static network model.

References \cite{jia2008competitions,duan2011competition,ileri2005demand,niyato2008dynamics,elias2010joint,sengupta2007economic,gao2010map,jia2009revenue} looked at the oligopoly issues, either between two operators \cite{jia2008competitions,duan2011competition} or among many operators \cite{ileri2005demand,elias2010joint,niyato2008dynamics,sengupta2007economic,gao2010map,jia2009revenue}.
For the case of two operators,  Jia and  Zhang in \cite{jia2008competitions} proposed a non-cooperative two-stage game model to study the duopoly competition. Duan et al.~ in \cite{duan2011competition} formulated the economic interaction among the spectrum owner, two secondary operators and the users as a three-stage game.
For the case of many operators,  Ileri et al.~ in \cite{ileri2005demand} developed a non-cooperative game to model competition of operators in a mixed commons/property-rights regime under the regulation of a spectrum policy server.  Elias and Martignon in \cite{elias2010joint} showed that polynomial pricing functions lead to unique and efficient Nash equilibrium for the two-stage Stackelberg game between network operators and secondary users. Niyato  et al.~ in \cite{niyato2008dynamics} formulated an evolutionary game for modeling the dynamics of a multiple-seller, multiple-buyer spectrum trading market. In addition, several auction mechanisms were proposed to study the investment problems of cognitive network operators (\eg\cite{sengupta2007economic,gao2010map,jia2009revenue}).

All results mentioned above considered a rather static network model. In contrast, our work adopts a dynamic network model to characterize the stochastic nature of wireless networks. We will focus on a monopoly model in this paper.

In this paper, we use Lyapunov stochastic optimization to show the optimality and stability of the proposed profit maximizing control algorithms.
%Lyapunov stochastic optimization was first applied to wireless networks in the seminal paper \cite{tassiulas1992stability}, and has been developed into a framework of dynamic algorithm design for general stochastic network optimization problems\cite{neely2010stochastic}.
Several closely related previous results applying Lyaunov stochastic optimization to wireless networks include  \cite{huang2010optimality,urgaonkar2009opportunistic,lotfinezhad2010optimal}. Huang and Neely in \cite{huang2010optimality} considered revenue maximization problem for a conventional wireless access point without considering the cognitive radio technologies. Urgaonkar and Neely in \cite{urgaonkar2009opportunistic} and Lotfinezhad et al.~ in \cite{lotfinezhad2010optimal} studied cognitive radio networks based on a user-oriented approach, by designing joint scheduling and resource allocation algorithms to maximize the utility of a group of secondary users. Our paper focused on an operator-oriented approach to address profit maximization problem. In particular, we need to deal with the combinatorial problem of channel selection and channel assignment that usually leads to a high computational complexity. By discovering and utilizing the special problem structure, we design a low-complexity algorithm that is suitable for online implementation.
%\com{what will be the additional technical challenge of looking at the operator's approach? Or it is just another molding with the same technical issues?} \res{how about emphasize we reduce the complexity by utilizing the structural information?}

\section{System Model}
\label{sec:system_model}

Consider a C-MVNO that provides wireless communications services to its own secondary users by acquiring spectrum resource from some spectrum owner.
For example, Google may acquire spectrum from AT$\&$T to provide its own wireless services through the C-MVNO model. The spectrum owner's spectrum can be divided into two types: the \emph{sensing band} and the \emph{leasing band}.
%, or utilizing the channels named as \emph{sensing band}.
In the sensing band, AT\&T serves its own primary users, but allows Google to identify available spectrum in this band through spectrum sensing without explicit communications with AT\&T. In the leasing band, AT\&T will does not allow spectrum sensing, and will lease the band to Google for economic returns.
% These sensing band channels are allocated to active primary users of the spectrum owner, but they are also open for spectrum sensing when the spectrum owner is not using them.

More specifically, we consider a time-slotted OFDM system, where the C-MVNO serves the downlink transmissions from its base station to the secondary users. The system model is illustrated in Fig.~\ref{fig:model}. Secondary users randomly arrive at the secondary network and request files with random sizes to be downloaded from the base station.  This requested files are queued at the server in the base station until they are successfully transmitted to the requesting users.

\begin{figure}[hbt]
\centering
\includegraphics[scale=0.35]{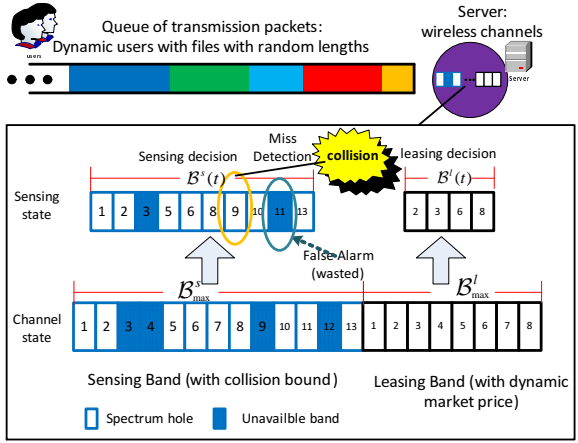}
\caption{Business model of the operator (Cognitive Virtual Network Operator).}
\label{fig:model}
\end{figure}

The rest of the section introduces each part of the system model in more details. The C-MVNO (or ``operator'' for simplicity) obtains wireless channels through spectrum sensing (Sections~\ref{sec:sensing} and \ref{sec:collision}) and spectrum leasing (Section~\ref{sec:leasing}), and allocates power over the obtained channels (Section~\ref{sec:power}). Secondary users dynamically arrive and request file downloading services (based on the demand model in Section~\ref{sec:demand_model}), and we model the requests as a queue (Section~\ref{sec:queuing}).

\subsection{Imperfect Spectrum Sensing}
\label{sec:sensing}
Sensing band $\mathcal{B}^s_{\max}\mdef\{1,\ldots,{B}^s_{\max}\}$ includes all channels that the spectrum owner allows sensing by the operator.\footnote{The operator will collect the sensing
information from a sensor network or geolocation database and provide it to its users, \ie  providing ``sensing as service''\cite{Weiss:2010p14709,geolocation}. This means that the network can accommodate legacy mobile devices without cognitive radio capabilities. For more detailed discussions, see\cite{duan2011cognitive}.}
%The state of a channel depends on the primary users' activities on this channel.
We define \textbf{the state of a channel $i\in\mathcal{B}^{s}_{\max}(t)$ in time slot $t$} as $S_{i}(t)$, which equals 0 if channel \ensuremath{i} is busy (being used by a primary user), and equals 1 if channel \ensuremath{i} is idle.
%
%For each channel $i\in\mathcal{B}^{s}_{\max}$,

We assume that $S_{i}(t)$ is an i.i.d. Bernoulli random variable, with an idle probability $p_{0}\in(0,1)$ and a busy probability $1-p_{0}$. This approximates the reality well if the time slots for secondary transmissions are sufficiently long or the primary transmissions are highly bursty \cite{anandkumar2010opportunistic}. (We will further study the general Markovian model in Section~\ref{sec:Markovian_model}.) We define the \textbf{{sensing state} of a channel $i\in\mathcal{B}^{s}_{\max}$ in time slot $t$} as $W_{i}(t)$, which equals to 0 if channel \ensuremath{i} is \emph{sensed} busy, and 1 if sensed idle.

{Notice that $W_{i}(t)$ may not equal to $S_{i}(t)$ due to imperfect sensing.
%
%Notice that the spectrum owner uses the sensing band to serve its own legacy primary users. To reduce signaling overheads and make the operation compatible with legacy primary systems that are not aware of secondary users, there are usually no direct communications between the (secondary) operator and spectrum owner in this band. Thus imperfect spectrum sensing can lead to false alarms and missed detections.
The accuracy of spectrum sensing depends on the sensing technology \cite{woyach2010can}.} If we denote $C^{s}$ as the sensing cost (per channel)\footnote{The cost corresponds to, for example, power or time used for sensing.}, then we can write the false alarm probability as $P_{fa}(C^s)\mdef Pr\{W_i=0|S_i=1\}$ (same for all channel $i$) and the missed detection probability as  $P_{md}(C^s)\mdef Pr\{W_i=1|S_i=0\}$ (same for all channel $i$).  Both functions are decreasing in $C^{s}$.
Intuitively, a better technology will have a higher cost $C^{s}$,  a lower false alarm probability $P_{fa}(C^s)$, and a lower missed detection probability $P_{md}(C^s)$. We denote all choices of cost $C^{s}$ (and thus the corresponding sensing technologies) by a finite set $\mc^s$.

As different channels have different conditions (to be explained in details in Section~\ref{sec:power}), the operator needs to decide which channels to sense at the beginning of each time slot. We use $\mathcal{B}^s(t)$ to denote \textbf{the set of channels sensed by the operator at time $t$}, which satisfies
\begin{equation}
\label{eq:sensing_bandwidth_con}
    \mathcal{B}^s(t)\subseteq\mathcal{B}^s_{\max},\forall\,t.
\end{equation}

\subsection{Collision Constraint}
\label{sec:collision}
Missed detections in spectrum sensing lead to transmission collisions with the primary users. We denote the \textbf{collision in channel $i\in\mathcal{B}^s_{\max}$ at time $t$} as a binary random variable $X_{i}(t)\in\{0,1\}$. We have  $X_i(t)\mdef (1-S_i(t))W_i(t)$, \ie  the collision happens if and only if the channel is busy but is sensed idle.
%As mentioned in the previous section, the missed detection probability $\alpha(C^s)\mdef Pr\{X_i(t)=1\}$} is an increasing function of the sensing

%\com{rewrite in terms of collision probability over the entire bands.}
To protect  primary users' transmissions, the operator needs to ensure that the average collision in each channel $i$ does not exceed a tolerable level $\eta_i$ (measured in terms of  the average  number of collisions per unit time) specified by the spectrum owner. The tolerable level $\eta_i$ can be channel specific, since the primary users in different channels may have different QoS requirements.
%\com{rewrite in terms of channel. Can save one variable.}
We define the time-average number of collision in channel $i$ as $\overline{X_i} \overset{\Delta}{=}\lim_{t\rightarrow \infty}\frac{1}{t}\sum_{\tau=0}^{t-1}\expectation\,[X_i(\tau)]$. The collision constraints are
\begin{equation}
    \overline{X_i}\le \eta_i, \forall i\in\mathcal{B}^s_{\max}(t).
\end{equation}

\subsection{Spectrum Leasing with a Dynamic Market Price}
\label{sec:leasing}
A spectrum owner may have some channels that do not want to be sensed, for either privacy reasons or the fear of collisions due to sensing errors. However, these channels may not be always fully utilized. The spectrum owner can lease the unused part of these channels to the operator dynamically over time to earn more revenue.
Recall that we denote the set of these channels as the  leasing band $\mb^l_{\max}\mdef\{1,\ldots,B^l_{\max}\}$. (In general, we may represent it as $\mb^l_{\max}(t)$, since our model allows leasing band to be time-varying. For the simplicity of notations, we denote it as $\mb^l_{\max}$ whenever it is clear.)
We use $\mb^l_i(t)$ to denote \textbf{the set of channels leased by the operator at time $t$}, which satisfies
\begin{equation}
\label{eq:leasing_bandwidth_con}
    \mb^l(t)\subseteq \mb^l_{\max},\forall\,t.
\end{equation}
%at the market price $C^l(t)$ per channel.
These channels will be exclusively used by the operator in the current time slot. We denote the leasing price per channel as $C^l(t)$, which stochastically changes according to the supply and demand relationship in the spectrum market (which might involve many spectrum owners and operators). It can be modeled by an exogenous (not affected by this particular operator's decisions) random process with countable discrete states and stationary distribution (not necessarily known by the operator).

\subsection{Power Allocation}
\label{sec:power}
In wireless network, there are usually channel fading due to  multipath propagation or shadowing from obstacles. To combat channel fading, it is necessary for the operator to do proper power allocation in both sensing channels and leasing channels to achieve satisfactory data rates.
For each channel $i\in\mb_{\max}\mdef \mb^s_{\max}\cup \mb^l_{\max}$, $h_i(t)$ represents its channel gain in time slot $t$ and follows an i.i.d. distribution over time.
Different channels have independent and possibly different channel gain distributions. We assume that secondary users are homogeneous and experience the same channel condition for the same channel. But channel conditions can be different in different channels.\footnote{This is the case where the users are located close by, and thus the downlink channel condition from the base station to the users is user independent.} (The heterogenous user scenario will be further discussed in Section~\ref{sec:MPMC_policy}.) The operator can measure $h_i(t)$ for each $i$ at the beginning of each slot $t$, but may not know the distributions.
%To guarantee the transmission performance, especially when the channel condition is bad, a smart power allocation is desired.
Let $P_i(t)$ denote the \textbf{power allocated to channel $i$ at time $t$}. Since we consider a downlink case here, the operator needs to satisfy the\textbf{ total power constraint} $P_{\max}$ at its base station,
\begin{equation}
\label{eq:power_con}
    \sum_{i\in\mb_{\max}}P_i(t)\le P_{\max}, \forall t.
\end{equation}
%\com{Always try to define symbols ($I_{i}(t)$ here) BEFORE we use it, unless it is something that the readers are very familiar of (and then we still need to define asap after its first use.}

In addition, for a channel $i\in\mathcal{B}^{s}_{\max}$ in the sensing band, we use the binary variable $I_{i}(t)\mdef S_i(t)W_i(t)$ to denote the \textbf{transmission result} of a secondary user, \ie $I_{i}(t)=1$ if successful (\ie $S_i(t)=1$ and $W_i(t)=1$)  and  $I_{i}(t)=0$ otherwise (either not sensed, or sensed busy, or sensed idle but actually busy). Based on the discussion of the leasing agreement, we have $I_i(t)\equiv 1$, $i\in\mb^{l}_{\max}$ for all channel in leasing band.
Then the rate in channel $i$ at time slot $t$ is (based on the  Shannon formula)
\begin{equation}
\label{eq:rate_fun}
    r_i(t)=  I_i(t)\log_2(1+h_i(t)P_i(t)),
\end{equation}
and \textbf{total transmission rate} obtained by the operator is
\begin{equation}
\label{eq:rate_expand}
    r(t)=\sum_{i\in\mb^s(t)\cap\mb^l(t)}r_i(t).
\end{equation}

Furthermore, we assume that the operator has a finite maximum transmission rate, \ie
$r(t)\le r_{\max}, \forall t$, under any feasible power allocation.

\subsection{Demand Model}
\label{sec:demand_model}
We will focus on elastic data traffic in this paper. Secondary users randomly arrive at the network to request files with random and finite file sizes (measured in the number of packets) from the operator. A user will leave the network once it has downloaded the complete requested file. The operator can price the packet transmission dynamically over time, which will affect the users' arrival rate.
%According to the demand and supply situations in the network, the operator make pricing strategies, which will further affect the demand from the user side,
For example, a higher price at peak time can refrain users from downloading files, as they can wait until a later time with a lower price. To model this, we use  $M(t)$ to denote the random \textbf{market state}, which can be measured precisely at the beginning of each time slot $t$ and can help estimate the users demand\footnote{For example, $M(t)$ can be users' willingness to pays, or whether the system is in peak time or off-peak time.}.  The random variable is drawn from a finite set $\mm$ over time in an i.i.d. fashion. The distribution of $M(t)$ may not be  known by the operator.

At a time $t$, the operator will decide whether to accept new file downloading requests from newly arrived secondary users. We define the \textbf{binary demand control variable as $O(t)$}, where $O(t)=1$ means that the operator accepts the incoming requests in time $t$, and $O(t)=0$ otherwise. When the operator decides to accept new requests of packet transmissions, it will also announce \textbf{a price $q(t)$ for transmitting one packet (to any user)}. This price will affect the users' incentives of downloading requests, \eg when price $q(t)$ is high, some users may choose to postpone their requests.
%Thus the current price directly affects the number of incoming users in the current slot.
%Users will respond to this price by adjusting their requests. More precisely,

More precisely, we denote the number of incoming users at time $t$ as a discrete random variable $N(t)\mdef N\left(M(t),q(t)\right)\in \{0,1,2\dots\}$, the distribution of which is a function of the transmission price $q(t)$ and market state $M(t)$.
Further, a user $n$'s requested file size is denoted $L_{n}(t)$, with $n\in\{1,2,\dots,N(q(t))\}$, which is assumed to be independent of each other and does not depend on $q(t)$ or $M(t)$. Moreover, we assume that users are using a set $\mathcal{K}=\{1,2,\dots,K\}$ of different applications, and denote $\theta_k$ as the probability that an incoming user is using application $k\in \mk$ with $\sum_{k=1}^{K} \theta_k=1$. The distributions of the file length for different applications can be different, and we denote $l_k$ as the expected file length of application $k\in \mk$.

To summarize, \textbf{users' {instantaneous} demand at time $t$} is
\begin{equation}
A(t)\mdef\sum_{n=1}^{N\left(M(t),q(t)\right)}L_n(t),
\end{equation}
which is a random variable {due to random file sizes and the random number of incoming users} (even given $q(t)$ and $M(t)$). We define the \textbf{users' {(expected)} demand function} as $D(t)\mdef D\left(M(t),q(t)\right)\mdef \expectation\,[A\left(M(t),q(t)\right)]$, and its value is completely determined by $M(t)$ and $q(t)$. We can calculate that $D(M(t),q(t))=\expectation\,[N(M(t),q(t))]\sum_{k\in\mk}\theta_kl_k$. Then it is reasonable to assume that the operator can rather accurately characterize the expected number of incoming users $\expectation\,[N(M(t),q(t))]$ through long-term observations. Thus the demand function $D(M(t),q(t))$ is known by the operator.
%\res{$A(t)$ is the instantaneous demand, which is a random variable even given the values of $q(t)$ and $M(t)$. $D$ is a deterministic function. When  }
We further assume that
the instantaneous demand is upper-bounded as $A(t)\le A_{\max}$ for all $t$, and that the demand function $D(t)$ is non-negative and non-increasing function of the price $q(t)$. When the price is higher than some upper-bound, \ie $q(t)\ge q_{\max}$, the demand function  $D(t)$ will be zero. The optimization of $O(t)$ and $q(t)$ based on the demand function will be further discussed in Section~\ref{sec:RevenueMaximization}.

\subsection{Queuing dynamics}
\label{sec:queuing}
Since we focus on the profit maximization problem in this paper, we will take a simple view of the network and model users' dynamic arrivals and departures as a single server queue.  When a user accesses the network, the corresponding file will be queued in a server at the base station, waiting to be transmitted to the user according to the First Come First Serve (FCFS) discipline.
% will be first queuing in the buffer of the system, and waiting for the transmission service according to a First In First Out (FIFO) order. %This simplification is valid, and coincides
%Though it may simplify some transmission details,
Shama and Lin in \cite{sharma2011ofdm} showed that the single server queue model is a good approximation for an OFDM system, especially when the number of users and channels are large.
%This model coincides with the result in : the OFDM system can be approximated as one-server queue.

We denote the \textbf{queue length} (\ie the backlog, or the number of all packets from all queued files) at time $t$ as $Q(t)$. Thus the queuing dynamic can be written as
\begin{equation}
\label{eq:queue}
    Q(t+1)=\Big(\!Q(t)-r(t)\Big)^++ O(t)A(t),
\end{equation}
where $(a)^+\mdef\max(a,0)$, $r(t)$ and $A(t)$ are the transmission rate and incoming rate at time $t$, and $O(t)$ is the binary demand control variable (\ie $O(t)=1$ means the operator admit the users' transmission requests at time $t$).  Throughout the paper, we
adopt the following notion of \emph{queue stability}:
\begin{equation}
    \overline{Q}\mdef\underset{t\rightarrow\infty}{\lim\sup} \frac{1}{t}\sum_{\tau=0}^{t-1}\expectation\,[Q(\tau)]<\infty.
\end{equation}

\section{Problem Formulation}
\label{sec:problem_formulation}
For notation convenience, we introduce several condensed notations and use them together with the original notations.

We define $\phi(t)\mdef(M(t),{\boldsymbol{h}}(t),{C^l}(t))$ as \textbf{observable parameters}, including the market state $M(t)$, channel conditions (vector) ${\boldsymbol{h}}(t)$, and the leasing price ${C^l}(t)$ in the spectrum market. Based on previous assumptions, $\phi(t)$'s are i.i.d over time and take values from a finite set $\Phi$.

We define $\gamma(t)\mdef(O(t),q(t),{C^s}(t), \mb^s(t), \mb^l(t), \boldsymbol{P}(t))$ as \textbf{decision variables}, including the demand control variable $O(t)$, the transmission price for users $q(t)$, the sensing cost (with the corresponding sensing technology) ${C^s}(t)$, the set of sensing channels $\mb^s(t)$, the set of leasing channels $\mb^l(t)$, and power allocations (vector) $\boldsymbol{P}(t)$ of the operator.
%resource budget of the sensing bandwidth, the leasing bandwidth and the power allocation of each channel.%
We assume that $\gamma(t)$ takes values form a countable (finite or infinite) set $\Gamma_{\phi}(t)$, which is a Cartesian product of the feasible regions of all variables, \ie non-negative values satisfying constraints (\ref{eq:sensing_bandwidth_con}), (\ref{eq:leasing_bandwidth_con}), and (\ref{eq:power_con}). With the condensed notations, functions in this paper can be simply represented as functions of $\gamma(t)$ with parameter $\phi(t)$.

We further define the \textbf{instantaneous profit in time $t$}
\begin{align}
\label{eq:instantaneous_profit}
    R(t)&\mdef R(\gamma(t);\phi(t))\nonumber\\
&\mdef q(t)O(t)A(t)-C^s(t)|\mb^s(t)|-C^l(t)|\mb^l(t)|.
\end{align}
%where the cost vector is $c_i(t)\mdef(C_i^s(t),C_i^l(t),C^P)$.
The \textbf{time average profit} is denoted as
$$\overline{R}\mdef{\underset{t\rightarrow\infty}{\lim\sup}\frac{1}{t}\sum_{\tau=0}^{t-1}\expectation\,[R(t)]}.$$  All expectations in this paper are taken with respect to system parameters $\phi(t)$ unless stated otherwise.

We look at the profit maximization problem through pricing determination and resource allocations. At the beginning of each time slot $t$, the operator observes the value of $\phi(t)$ and makes a decision $\gamma(t)$ to maximize the time average profit, subject to the system stability constraint (\ref{eq:queue_stable}) and the collision upper-bound requirement (\ref{eq:collision_con}). The Profit Maximization (PM) problem is formulated as
\begin{align}
\text{\textbf{PM}:}\;\;\text{Maximize}\;\;\;& \overline{R}\nonumber\\
\text{Subject to}\;\;\;& \overline{Q}<\infty,\label{eq:queue_stable}\\
&    \overline{X_i}\le \eta_i, i\in\mb^s_{max},\label{eq:collision_con}\\
\text{Variables}\;\;& \gamma(t)\in \Gamma_{\phi(t)}, \forall t,\nonumber\\
\text{Parameters}\;\;&\phi(t),\forall t.\nonumber
\end{align}
We represent its optimal solution as $\gamma^{*}(t)=\left(O^{*}(t),q^{*}(t),C^{s*}(t), \mb^{s*}(t), \mb^{l*}(t), \boldsymbol{P}^*(t)\right)$, and denote $\overline{R}^*$ as the maximum profit. {The PM problem is an infinite horizon stochastic optimization problem, which is in general hard to solve directly, especially when the distribution of dynamic parameter $\phi(t)$ is  unknown. For example, the future leasing price is hard to predict due to the dynamic supplies and demands in the market; and the primary users' activities can not be estimated precisely before hand.}

\section{Profit Maximization Control Policy}
\label{sec:PMC_policy}
Now, we adopt Lyapunov stochastic optimization technique to solve the PM problem.

\subsection{Lyapunov stochastic optimization}
We first introduce a virtual queue for constraint \eqref{eq:collision_con}, and then derive the optimal control policy to solve the PM problem through the technique of drift-plus-penalty function minimization \cite{neely2010stochastic}.

We denote $Z_i(t)$ as the number of collisions happening in sensing channel $i\in\mb^s_{\max}$.  The counter $Z_i(t)$ can be understood as a ``virtual queue'', in which the incoming rate is $X_i(t)$, and the serving rate is $\eta_i$ (the collision tolerant level).  The queue dynamic is
\begin{equation}
\label{eq:virtual_queue}
    Z_i(t+1)=\Big(Z_i(t)-\eta_i\Big)^++ X_i(t),
\end{equation}
with $Z_i(0)=0$.
By this notion, if the virtual queue is stable, then it implies that the average incoming rate is no larger than the average serving rate. This is just the same as the collision upper-bound constraint (\ref{eq:collision_con}).

We introduce the general queue length vector $\boldsymbol{\Theta}(t)\mdef\{Q(t), \boldsymbol{Z}(t)\}$.
We then define the Lyapunov function
\begin{equation*}
    L(\boldsymbol{\Theta}(t))\overset{\Delta}{=}\frac{1}{2}\Big[Q(t)^2+\sum_{i\in\mb^s_{\max}}Z_i(t)^2\Big],
\end{equation*}
and the Lyapunov drift
\begin{equation}
\label{eq:Lyapunov_drift}
\Delta(\boldsymbol{\Theta}(t))\overset{\Delta}{=}\expectation\,\Big[{L(\boldsymbol{\Theta}(t+1))-L(\boldsymbol{\Theta}(t))|\boldsymbol{\Theta}(t)}\Big].
\end{equation}

According to the Lyapunov stochastic optimization technique, we can obtain instantaneous control policy that can solve the PM problem though minimizing {some upper bound of} the following drift-plus-penalty function in every slot $t$:
\begin{equation}
\label{eq:drift_penality}
    \Delta(\boldsymbol{\Theta}(t))-V\expectation\,[R(t)|\boldsymbol{\Theta}(t)].
\end{equation}
There are two terms in the above function. The first term is the Lyapunov drift defined in \eqref{eq:Lyapunov_drift}. It is shown by Lyapunov stochastic optimization \cite{neely2010stochastic} that we can achieve the system stabilities (\ie constraints \eqref{eq:queue_stable} and \eqref{eq:collision_con} of the PM problem) by {showing the existence of a constant upper bound for the drift function.} The second term in \eqref{eq:drift_penality} is just the objective of the PM problem, \ie to minimize the minus profit, which is equivalent to maximize the profit. Here parameter $V$ is introduced to achieve the desired tradeoff between profit and queuing delay in the control policy.
{We first find an upper bound for (\ref{eq:drift_penality}).}

By the queue dynamic (\ref{eq:queue}), we have
\begin{align}
    &Q(t\!+\!1)^2\le \left(Q(t)\!-\!r(t)\right)^2\!+\!A(t)^2\! +\! 2Q(t)O(t)A(t)\nonumber\\
&\quad= Q(t)^2\!+\!r(t)^2 +A(t)^2\!+\! 2Q(t)\!\left(O(t)A(t)\!-\!r(t)\right).
\label{eq:queue_expand}
\end{align}
Similarly, for virtual queue (\ref{eq:virtual_queue}), we have
\begin{align}
\label{eq:virtual_queue_expand}
Z_i(t\!+\!1)^2\le Z_i(t)^2\!+ \!\eta_i^2\!+\!X_i(t)^2\!+\!2Z_i(t)(X_i(t)\!-\!\eta_i).
\end{align}

\rev{
Substituting (\ref{eq:queue_expand}) and (\ref{eq:virtual_queue_expand}) into  (\ref{eq:drift_penality}), we have
%\begin{align}
%\label{eq:drift_penalty_expand}
%&\Delta(\boldsymbol{\Theta}(t))-V\expectation\,[R(t)|\boldsymbol{\Theta}(t)]\le D - V\expectation\,[R(t)|\boldsymbol{\Theta}(t)] \nonumber\\
%&\quad\quad+Q(t)\expectation\,\left[O(t)A(t)-r(t)|\boldsymbol{\Theta}(t)\right] \nonumber\\
%&\quad\quad+\sum_{i\in\mb^s_{\max}}Z_i(t)\expectation\,\left[X_i(t)-\eta_i|\boldsymbol{\Theta}(t)\right]
%\end{align}
\begin{align*}
%\label{eq:drift_penalty_expand}
&\Delta(\boldsymbol{\Theta}(t))-V\expectation\,[R(t)|\boldsymbol{\Theta}(t)]\le D-D_1(t)\nonumber\\
&-V\expectation\left[q(t)O(t)A(t)\!-\!C^s(t)|\mb^s(t)|\!-\!C^l(t)|\mb^l(t)||\boldsymbol{\Theta}(t)\right]\nonumber\\
&+Q(t)\expectation\left[O(t)A(t)-r(t)|\boldsymbol{\Theta}(t)\right] +\sum_{i\in\mb^s_{\max}}\!Z_i(t)\expectation\left[X_i(t)|\boldsymbol{\Theta}(t)\right]
\end{align*}
where $D$ is a positive constant satisfying the following condition for all $t$,
\begin{align*}
    D\!\ge\!\frac{1}{2} \!\left(\!\expectation\Big[r(t)^2\!+\! (O(t)A(t))^2|\boldsymbol{\Theta}(t)\Big]\!\!  + \!\!\!\!\!\!\sum_{i\in\mb^s_{\max}}\!\!\!\!\!\expectation\Big[X_i(t)^2\!+\! \eta_i^2|\boldsymbol{\Theta}(t)\Big]\!\!\right)\!\!\!\!,
\end{align*}
%
%We further expand the right hand side of (\ref{eq:drift_penalty_expand}):
%\begin{align*}
%&\Delta(\boldsymbol{\Theta}(t))-V\expectation\,[R(t)|\boldsymbol{\Theta}(t)]\le D-D_1(t)\nonumber\\
%&-V\expectation\left[q(t)O(t)A(t)\!-\!C^s(t)|\mb^s(t)|\!-\!C^l(t)|\mb^l(t)||\boldsymbol{\Theta}(t)\right]\nonumber\\
%&+Q(t)\expectation\left[O(t)A(t)-r(t)|\boldsymbol{\Theta}(t)\right] \nonumber\\
%&+\sum_{i\in\mb^s_{\max}}Z_i(t)\expectation\left[X_i(t)|\boldsymbol{\Theta}(t)\right]
%\end{align*}
and $D_1(t)\mdef\sum_{i\in\mb^s_{\max}} Z_i(t)\eta_i$
 is a known constant at time $t$, since the values of $Z_i(t)$s are known  at time $t$. %\com{from equation (18) to here: need to explain why we do these calculations. }

%Note that the collision between secondary and primary users can only happen in the  channels that have been chosen for sensing, \ie  $i\in\mb^s(t)$. By this fact and (\ref{eq:rate_expand}), we have
%\com{physical meaning of the equation below?}
%\begin{align*}
%&\Delta(\boldsymbol{\Theta}(t))-V\expectation\,[R(t)|\boldsymbol{\Theta}(t)]\le D-D_1(t)\nonumber\\
%&-V\expectation\left[q(t)O(t)A(t)\!-\sum_{i\in\mb^s(t)}C^s(t)-\sum_{i\in\mb^l(t)}C^l(t)|\boldsymbol{\Theta}(t)\right]\nonumber\\
%&+Q(t)\expectation\left[O(t)A(t)-\sum_{i\in\mb^s(t)\cup \mb^l(t)}r_i(t)|\boldsymbol{\Theta}(t)\right] \nonumber\\
%&+\sum_{i\in\mb^s(t)} Z_i(t)\expectation\left[X_i(t)|\boldsymbol{\Theta}(t)\right]
%\end{align*}
To further simplify the above expression,
we introduce two new notations: channel selection $\mb(t)\mdef\mb^s(t)\cup \mb^l(t)$, and channel cost
\begin{equation}\label{eq:virtual_cost}
C_i(t)\!\mdef\!
\begin{cases}
\widetilde{C^s}(t),&\text{if }i\in\mb^s(t), \\
C^l(t),      &\text{if }i\in\mb^l(t),
\end{cases}
\end{equation}
where $\widetilde{C^s}(t)$ is the virtual sensing cost and is defined as $\widetilde{C^s}(t)\mdef C^s(t)+(1/V)Z_i(t)\expectation\left[X_i(t)|\boldsymbol{\Theta}(t)\right]$. %\com{explain the intuitions of the additional part of the virtual sensing cost.}
Note that this virtual sensing cost depends not only on the sensing cost but also on the collision history in this channel. More frequent past collisions in this channel will increase the virtual sensing cost, hence makes the operator more conservative about choosing this sensing channel.
It then follows:
\begin{align}
\label{eq:pmc}
&\Delta(\boldsymbol{\Theta}(t))-V\expectation\,[R(t)|\boldsymbol{\Theta}(t)]\le D-D_1(t)\nonumber\\
&\quad\quad\quad+V\expectation\left[\left.\left(\frac{Q(t)}{V}-q(t)\right)O(t)A(t)\right|\boldsymbol{\Theta}(t)\right]\nonumber\\
&\quad\quad\quad+V\expectation\left[\left.\sum_{i\in\mb(t)} \!C_i(t)-\frac{Q(t)r_i(t)}{V}\right|\boldsymbol{\Theta}(t)\right],
\end{align}
where we use the fact that collisions between secondary and primary users can only happen in  channels that are chosen for sensing, \ie  $i\in\mb^s(t)$.
%
%It then follows:
%\begin{align}
%\label{eq:pmc}
%&\Delta(\boldsymbol{\Theta}(t))-V\expectation\,[R(t)|\boldsymbol{\Theta}(t)]\le D-D_1(t)\nonumber\\
%&\quad\quad\quad+V\expectation\left[\left.\Big(1/V)Q(t)-q(t)\Big)O(t)A(t)\right|\boldsymbol{\Theta}(t)\right]\nonumber\\
%&\quad\quad\quad+V\expectation\left[\left.\sum_{i\in\mb(t)} \!C_i(t)-(1/V)Q(t)r_i(t)\right|\boldsymbol{\Theta}(t)\right].
%\end{align}
Next we propose the Profit Maximization Control (PMC) policy to minimize the right hand side of inequality (\ref{eq:pmc}) for each time $t$.
}

\subsection{Profit Maximization Control (PMC) policy}
\label{sec:pmc}
It is clear that minimizing the right hand side of \eqref{eq:pmc} is equivalent to minimizing the last two terms in \eqref{eq:pmc}. Note that the last two terms are decouple in decision variables, thus we have the two parallel parts in the PMC policy as follows:

\subsubsection{Revenue Maximization} \label{sec:RevenueMaximization}

Here we determine two variables: the transmission price $q(t)$ and the market control decision $O(t)$. The optimal transmission price $q(t)$ is obtained by solving the following revenue maximization problem:
\begin{align}
\text{Maximize}\;& q(t)D\left(q(t),M(t)\right)\!-\!\frac{Q(t)}{V}D\left(q(t),M(t)\right)\label{eq:pricing}\\
\text{Variables}\;& q(t)\ge 0\nonumber
\end{align}

To obtain the above problem formulation of revenue maximization, we use the fact that the demand function $D(M(t),q(t))\mdef \expectation\,[A(t)]$, which is independent of the queuing states of the system.

Note that the first term in (\ref{eq:pricing}) is just the revenue that the operator collects from its users. The second term can be viewed as a shift of the queuing effect, which is introduced by the Lyapunov drift for system stability.

If the maximum objective in (\ref{eq:pricing}) (under the optimal choice of $q(t)$) is positive, the operator sets the demand control variable $O(t)=1$ and accepts the present incoming requests $A(t)$ at the price $q(t)$. Otherwise, the operator sets $O(t)=0$ and rejects any new requests.

\subsubsection{Cost Minimization}
 We determine channels selection $\mb(t)$, sensing technology (or cost) $C^s(t)$, and power allocation $\boldsymbol{P}(t)$, by solving the following optimization problem to control the costs of the operator to provide transmission services to its users.
\begin{align}
\text{Minimize}& \sum_{i\in \mb(t)}\!\! C_i(t)- \frac{Q(t)}{V}\expectation\,[r_i(t)|\boldsymbol{\Theta}]\label{eq:r_problem}\\
\text{Subject to}&\quad (\ref{eq:sensing_bandwidth_con}), (\ref{eq:leasing_bandwidth_con}), (\ref{eq:power_con}) \nonumber\\
\text{Variables}&\quad C^s(t),\mb^s(t),\mb^l(t), P_i(t)\ge 0\nonumber
\end{align}

To obtain the above problem formulation of cost minimization, we use the fact that $X_i(t)=(1-S_i(t))W_i(t)$, which is independent of the queuing state. Thus the virtual sensing cost can be updated as $\widetilde{C^s}(t)=C^s(t)+(1/V)Z_i(t)(1-p_0)P_{md}(C^s(t))$, which increases with the virtual queue and missed detection probability.

Note that the first term in the summation in \eqref{eq:r_problem} is the cost of each channel. The second term in the summation is a queuing-weighted expected transmission rate, again is a shift introduced by Lyapunov drift for system stability. This shift can be also viewed as the ``gain'' collected from the channel to help clear the queue.

\subsubsection{Intuitions behind the PMC policy}
We discuss some intuitions behind the PMC policy. To maximize the profit, the operator needs to perform revenue maximization and cost minimization. To guarantee the queuing stability, some shifts (\ie all queue-related terms) are introduced by the Lyapunov drift in these problems. {In Appendix Section~\ref{sec:queuing_effect}}, we show that the queueing effect will increase the optimal price announced by the operator (comparing with not considering queueing), as a higher price will reduce the users' demands and maintain the system stability.

The Lyapunov stochastic optimization approach provides a way to decompose a long-term average goal (\eg the PM problem) into instantaneous optimization problems (\eg revenue maximization and cost minimization problems in the PMC policy). In the stochastic optimization problem, the current decisions always have impacts on the future problems. These impacts are characterized and incorporated by the queueing shift terms  in the instantaneous optimization problems. Therefore, we can achieve the long term goal through focusing on the instantaneous decisions in every time slot.  The flowchart for the PMC policy is illustrated in Fig.~\ref{fig:flowchart}.
\begin{figure}[hbt]
\centering
\includegraphics[scale=0.43]{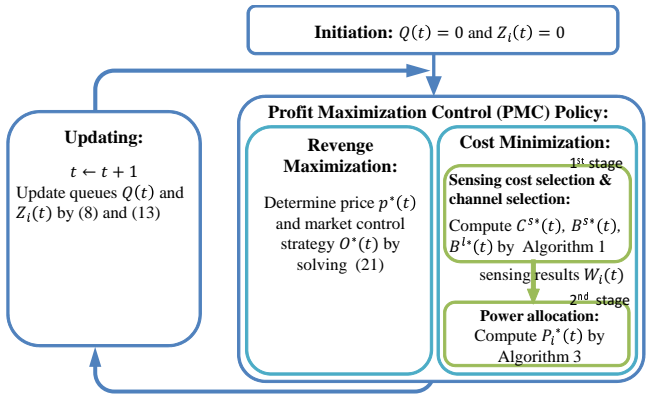}
\caption{Flowchart of the dynamic PMC policy}
\label{fig:flowchart}
\end{figure}
%
%In the PMC policy, the revenue maximization problem is a simple one-dimensional optimization problem of $q(t)$, and $O(t)$ is further determined by the optimal value of $q(t)$. We will further discuss the impact of queuing on the pricing strategy in Section~\ref{sec:queuing_effect}.

{Although the revenue maximization problem is relatively easy to solve,} the cost minimization problem is very complicated. It is actually a two-stage decision problem. In the first stage, the operator determines the sensing technology, and chooses which channels to sense and which channels to lease. Then spectrum sensing is performed to identify available channels. With this information, the operator further allocates downlink transmission power in the available channels (sensed idle ones and leasing ones). In Section~\ref{sec:cost_min}, we focus on designing algorithms to solve the cost minimization problem.

\subsection{Algorithms for Cost Minimization Problem}
\label{sec:cost_min}
Now we use backward induction to solve the cost mimmation problem.
\subsubsection{The Second Stage Problem}
\label{sec:2nd_stage}
We first analyze the power allocation in the second stage, where the sensing results $W_i(t)$, the channel selection $\mb^s(t)$, and the sensing technology $C^s(t)$  have been determined. Therefore, the power allocation problem of  (\ref{eq:r_problem}) is as follows:
\begin{align}
\text{Maximize}&\quad \sum_{i\in \mb(t)} \omega_i(t)\log \left(1+h_i(t)P_i(t)\right)\label{eq:power_allocation_problem}\\
\text{Subject to}&\quad     \sum_{i\in\mb_{\max}}P_i(t)\le P_{\max}, \nonumber\\
\text{Variables}&\quad P_i(t)\ge 0\nonumber
\end{align}
where
\begin{equation}
\label{eq:omega}
\omega_i(t)\mdef
  \begin{cases}
   \omega_0\mdef\expectation\,[S_i(t)|W_i(t)=1],&\text{if }i\in\mb^s(t), \\
   1,       &\text{if }i\in\mb^l(t).
  \end{cases}
\end{equation}
and we can calculate
\begin{align}
&\expectation\,[S_i(t)|W_i(t)=1]\nonumber\\
&=\frac{p_0(1-P_{fa}(C^s(t)))}{p_0(1-P_{fa}(C^s(t)))+(1-p_0)P_{md}(C^s(t))}.
\end{align}
By using the Lagrange duality theory, we can show that the problem (\ref{eq:power_allocation_problem}) has the following optimal solution
\begin{equation}
\label{eq:power_allocation}
P_i^*(t)=
\begin{cases}
\omega_i(t)\left(\frac{1}{\lambda(t)}-\frac{1}{\omega_i(t)h_i(t)}\right)^+& i\in\mb(t),\\
0 & i\notin\mb(t),
\end{cases}
\end{equation}
where $\lambda(t)$ is the Lagrange multiplier of the total power constraint (\ref{eq:power_con}). The optimal value of $\lambda(t)$ is the following water filling solution,
\begin{equation}
\label{eq:dual_sol}
\lambda(t)= \frac{\sum_{i\in \mb_{p}(t)}\omega_i(t)}{P_{\max}+\sum_{i\in\mb_{p}(t)}\frac{1}{h_i(t)}},
\end{equation}
where $\mb_{p}(t)\mdef\{i\in\mb(t):P_i(t)>0\}$. Note that (\ref{eq:dual_sol}) is a fixed-point equation of $\lambda(t)$, and the precise value of $\lambda(t)$ is not given here.

%%%%%%%%%%%watefilling for power allocation%%%%%%%%%%%%%%%%%%%%%%%%%%%%%%%%%%%%
When the values of all parameters (\ie $h_i(t),\omega_i(t)$) are given, we can use a simple water level searching Algorithm~\ref{alg:power_allocation}  and similar as the searching algorithms in \cite{huang2009downlink,palomar2005practical}) to determine the exact optimal value of $\lambda(t)$.
In the following pseudo code of Algorithm~\ref{alg:power_allocation}, we define a function $\Lambda(m)$ as follows:
$$\Lambda(m)\mdef \frac{\sum_{i=1 }^m\omega_i}{P_{\max}+\sum_{i=1}^m\frac{1}{h_i}}.$$

\begin{algorithm}[hbt]
\caption{Power Allocation}
\label{alg:power_allocation}
\begin{algorithmic}[1]
%\Procedure{Computing $\lambda$} {$C^s$}
\State Rearrange the channel indices $i\in\mb_{\max}$ as a decreasing order of $\omega_i(t)h_i(t)$
\State $m\gets |\mb(t)|$, $\lambda\Leftarrow\Lambda(m)$ %\% Initialization
\While{$\lambda\ge h_m(t)\omega_m(t)$}
	\State $m\gets m-1$		
	\State $\lambda\Leftarrow\Lambda(m)$
\EndWhile
%\EndProcedure
\end{algorithmic}
\end{algorithm}

The main complexity in this algorithm is to sort the channels according to the channel gains. We can adopt established sorting algorithms \cite{knuth1973art} to obtain the index rearrangement with a complexity $\mathcal{O}(|\mb_{\max}|\log(|\mb_{\max}|))$. Thus the total complexity of Algorithm~\ref{alg:power_allocation}  is $\mathcal{O}(|\mb_{\max}|^3\log(|\mb_{\max}|))$.

\subsubsection{The First Stage Problem}
Let us consider the first stage problem to determine the sensing technology, the sensing set, and the leasing set. Note that since the sensing has not been performed at this stage yet, thus the sensing result $W_i(t)$ is not known.
We denote
\begin{equation}
    \alpha_i(t)=\begin{cases}
\alpha_0\mdef\expectation\,[S_i(t)W_i(t)]& \text{if } i\in\mb^s(t)\\
1 & \text{if } i\in\mb^l(t),
\end{cases}
\end{equation}
and we can calculate
\begin{align}
&\expectation\,[S_i(t)W_i(t)]=p_0(1-P_{fa}(C^s(t)))
\end{align}

Substitute the optimal power allocation (\ref{eq:power_allocation}) into the problem (\ref{eq:r_problem}), we have
\begin{align}
\text{Minimize}& \sum_{i\in \mb(t)}\!\! C_i(t)\!-\! \frac{Q(t)}{V}\alpha_i(t)\!\left(\!\log\!\left(\frac{h_i(t)\omega_i(t)}{\lambda(t)}\!\right)\!\!\right)^+\label{eq:r_problem_2}\\
\text{Subject to}&\quad \mb^s(t)\subseteq \mb^s_{\max},\; \mb^l(t)\subseteq \mb^l_{\max},\; C^s(t)\in\mc \nonumber\\
\text{Variables}&\quad\mb^s(t),\mb^l(t), C^s(t)\nonumber
\end{align}
We first consider the above problem for a fixed sensing cost $C^s(t)$. This problem is a combinatorial optimization problem of $\mb^s(t)$ and $\mb^s(t)$. The worst case of searching complexity (\ie exhaustive searching) can be $\mathcal{O}(2^{|\mb_{\max}|})$, exponential in the number of total channels.

However, we can reduce the complexity by exploring the special structure of this problem.

\begin{proposition}
\label{pro:threshold}
(Threshold Property)
\begin{itemize}
  \item We rearrange the leasing channel indices $i\in\mb^l_{\max}$ in the decreasing order of $g_i(t)$, which is defined as
\begin{equation}
\label{eq:virtual_channel_gain_l}
    g_i(t)\mdef h_i(t)\exp\left(-\frac{C_i(t)}{\frac{Q(t)}{V}}\right).
\end{equation}
%For all leasing channels $i\in\mb^l_{\max}$,
There exists a threshold index $i^l_{th}$, such that a channel $i$ is chosen for leasing (\ie $i\in\mb^l(t)$) if and only if $i\le i^l_{th}$.
  \item We rearrange the sensing channel indices $j\in\mb^s_{\max}$ in the decreasing order of $g_j(t)$, which is defined as
\begin{equation}
\label{eq:virtual_channel_gain_s}
    g_j(t)\mdef \omega_0h_j(t)\exp\left(-\frac{C_j(t)}{\frac{Q(t)}{V}\alpha_0}\right).
\end{equation}
For all leasing channels $j\in\mb^s_{\max}$, there exists a threshold index $j^s_{th}$, such that a channel $j$ is chosen for sensing (\ie $j\in\mb^s(t)$) if and only if $j\le j^s_{th}$.
\end{itemize}
\end{proposition}
%The proof of the above proposition is given in Appendix~\ref{sec:threshold_proposition}.
\begin{proof}
For each channel in the optimal channel selection set $i\in{\mb^*}(t)$, it satisfies the following condition
\begin{equation}
\label{eq:set_selection_sol}
C_i(t)\!\le\!\frac{Q(t)}{V}\!\alpha_i(t)\!\!\left(\log\left(\frac{\omega_i(t)h_i(t)}{\lambda(t)}\right)\right)^+.
\end{equation}
This result is easy to see from the objective function in (\ref{eq:r_problem_2}):
%for each channel $i$,
to optimize the profit, we should only pick the channel with its cost no larger than its gain.
Thus by \eqref{eq:virtual_channel_gain_l} and \eqref{eq:virtual_channel_gain_s}, the optimization problem in (\ref{eq:r_problem_2}) can be written in the following equivalent form:
\begin{align}
\text{Maximize}&\sum_{i\in \mb^l(t)}\!\!\!\! \left(\!\!\log\!\left(\!\frac{g_i(t)}{\lambda(t)}\!\right)\!\right)^+ \!\!\!\! + \!\alpha_0\!\!\!\! \!\!\sum_{j\in \mb^s(t)}\!\!\!\!\left(\!\!\log\!\left(\frac{g_j(t)}{\lambda(t)}\right)\!\right)^+\label{eq:r_problem_3}\\
\text{Subject to}&\quad \mb^s(t)\subseteq \mb^s_{\max},\quad \mb^l(t)\subseteq \mb^l_{\max} \nonumber\\
\text{Variables}&\quad C^s(t),\mb^s(t),\mb^l(t)\nonumber
\end{align}
Thus by the $\log$ function in the objective of \eqref{eq:r_problem_3}, the threshold property immediately follows.
\end{proof}

This proposition suggests that we should select the channel with a large $g_i$ (for leasing channels) or $g_j$ (for sensing channels). Note that as defined in \eqref{eq:virtual_channel_gain_l} and \eqref{eq:virtual_channel_gain_s}, $g_i$ and $g_j$ are equal to channel information (\ie $h_i$ for leasing channels, $\omega_jh_j$ for sensing channels) multiplying a decaying factor related to the channel cost.  They can be understood as \emph{virtual channel gains} by taking channel costs into consideration. A large value of $g_i$ or $g_j$ means that the channel is cost-effective, \ie the channel has a good channel gain as well as a low cost.

\rev{By Proposition~\ref{pro:threshold},  it is clear that we can obtain the optimal channel selection by an exhaustive search of the optimal sensing and leasing thresholds. Algorithm~\ref{alg:optimal_channel} gives a pseudo code for the searching procedure.}
\begin{algorithm}[ht]
\caption{Optimal Channel Selection (for a given $C^s(t)$) }
\label{alg:optimal_channel}
\begin{algorithmic}[1]
\Procedure {Computing $\mb$} {$C^s(t)$}
    \State \textbf{invoke} procedure SearchingThreshold($C^s$) to calculate $\Upsilon^l$ and $\Upsilon^s$
	\State $U(C^s(t))\gets 0$
    \For{$i=0,1\dots,\Upsilon^l$}
			\For{$j=0,1\dots,\Upsilon^s$}
				\State Calculate $\lambda(t)$ as (\ref{eq:dual_sol}) with $\mb_p(t)=\mb^l_i\cup\mb^s_j$
				\If {$g_i(t)>\lambda(t)$ and $g_j(t)>\lambda(t)$}
					\State Calculate $U(i,j)$
					\If{$U(i,j)<U(C^s(t))$}
						\State $U(C^s(t))\gets U(i,j)$
						\State $\mb(C^s(t))\gets \mb^l_i\cup\mb^s_j$
					\EndIf
				\EndIf
			\EndFor
    \EndFor
\EndProcedure
\end{algorithmic}
\end{algorithm}

In Algorithm \ref{alg:optimal_channel},  $U(i,j)$ denotes the optimal value of (\ref{eq:r_problem_3}) with the channel selection set $\mb=\mb^l_i\cup\mb^s_j$.
To decrease the number of searching loops, we can first run Algorithm~\ref{alg:water-filling} \rev{(in Appendix~\ref{sec:search_threshold})} to determine the maximum possible thresholds $\Upsilon^l$ for leasing channels or $\Upsilon^s$ for sensing channels. (If we do not run Algorithm~\ref{alg:water-filling} , we can just set $\Upsilon^l=|\mb^l_{\max}|$ and $\Upsilon^s=|\mb^s_{\max}|$. Whether we run Algorithm~\ref{alg:water-filling} or not, the complexity of Algorithm \ref{alg:optimal_channel} is no worse than $\mathcal{O}(|\mb^s_{\max}|\times |\mb^l_{\max}|)$.) Thus the searching complexity is reduced to $\mathcal{O}(|\mb_{\max}|^2)$, given the channel indices are rearranged as in the Proposition~\ref{pro:threshold}.  We can adopt established sorting algorithms \cite{knuth1973art} to obtain the index rearrangement with a complexity $\mathcal{O}(|\mb_{\max}|\log(|\mb_{\max}|))$. Thus the total complexity of finding the optimal channel selection is $\mathcal{O}(|\mb_{\max}|^3\log(|\mb_{\max}|))$.

%
%\begin{algorithm}[hbt]
%\caption{Search Threshold $\Upsilon^l$ (or $\Upsilon^s$) for a given $C^s(t)$}
%\label{alg:water-filling}
%\begin{algorithmic}[1]
%\Procedure{SearchingThreshold} {$C^s$}
%\State Rearrange the channel indecent $i\in\mb^l_{\max}$ (or $i\in\mb^s_{\max}$) as the decreasing order of $g_i(t)$.
%\State  $m\gets |\mb^l_{\max}|$ (or  $m\gets |\mb^s_{\max}|$), $\lambda\Leftarrow\Lambda(m)$ %\% Initialization
%\While{$\lambda\ge g_m(t)$}
%	\If {$m>1$}
%	\State $m\gets m-1$		
%	\State $\lambda\Leftarrow\Lambda(m)$
%	\Else \State   \bf{break}
%	\EndIf
%\EndWhile
%		\State $\Upsilon^l\gets m$ (or $\Upsilon^s\gets m$)
%\EndProcedure
%\end{algorithmic}
%\end{algorithm}

%\com{Since the complexity is high, need to add some discussions on how to reduce the complexity in practice. For example, would we have significant performance loss if we do this over several time slots instead of over each time slot?}

{Note that in real systems, the channel conditions and the leasing cost may not change as frequently as every time slot. We usually can update these network parameters every time frame (which is composed by several time slots instead of one time slot). Accordingly, the above algorithm will also be operated based on the time frames, which will greatly reduce the computation complexity in practice.}

Furthermore, let us find the optimal sensing cost $C^s(t)$ by enumerating all possible sensing costs $C^s(t)\in\mc^s$. For the sensing cost $C^s(t)$, we denote the objective value in (\ref{eq:r_problem_3}) as $U(C^s(t))$ and the optimal channel selection set as $\mb(C^s(t))$. The corresponding pseudo code is given in Algorithm \ref{alg:sensing_cost_search}, the complexity of which is $\mathcal{O}(|\mathcal{C}|\times|\mb_{\max}|^3\log(|\mb_{\max}|))$.

\begin{algorithm}[ht]
\caption{Optimal Sensing Cost and Channel Selection}
\label{alg:sensing_cost_search}
\begin{algorithmic}[1]
	\State $U^*\gets 0$
    \For{$C^{s}(t)\in\mathcal{C}^s$}
		\State Determine the optimal channel selection $\mb(C^s(t))$ (see Algorithm~\ref{alg:optimal_channel})
		\State Calculate $U(C^s(t))$
			\If{$U^*> U$}
				\State $U^*\gets U$, $C^{s*}(t)\gets C^s(t)$, $\mb^{*}(t)\gets \mb(C^s(t))$
			\EndIf
    \EndFor
\end{algorithmic}
\end{algorithm}

So far, we have completely solved the two-stage optimization problem in (\ref{eq:r_problem}). For each time $t$, the operator first runs Algorithm \ref{alg:sensing_cost_search} to choose the channel sets $\mb^{*}(t)=\mb^{s*}(t)\cup\mb^{l*}(t)$. Then it uses the sensing technology with a cost $C^{s*}(t)$ to sense channels in $\mb^{s*}(t)$. Based on the sensing results, it further runs Algorithm \ref{alg:power_allocation} to determine the power allocation $P_i^*(t),i\in\mb^{*}(t)$.

\subsubsection{Sensing vs. Leasing}
We are also interested in how the PMC policy makes the best tradeoff between sensing and leasing based on the sensing cost $C^s(t)$ and the leasing cost $C^l(t)$.
To make the comparison easy to understand, we will consider perfect sensing with no sensing errors (\ie $\omega_0=1$ and $\alpha_0=p_0$). We will further assume that a leasing channel $i$ and a sensing channel $j$ have the same channel gain $h_i(t)=h_j(t)$. Finally, {we assume that two channels have the same availability-price-ratio,} \ie the costs satisfy $C^s(t)=\alpha_0 C^l(t)$. We want to answer the following question: \emph{is the PMC policy indifferent in choosing either of the two channel?}

{By \eqref{eq:virtual_channel_gain_l} and \eqref{eq:virtual_channel_gain_s}, we have $g_i=g_j$ for these two channels.
By (\ref{eq:r_problem_3}), we can calculate the net gains by channel $i$ and $j$: $\left(\log\left(\frac{g_i(t)}{\lambda(t)}\right)\right)^+\ge \alpha_0\left(\log\left(\frac{g_j(t)}{\lambda(t)}\right)\right)^+$. To maximize the objective in (\ref{eq:r_problem_3}), it is clear that PMC policy will prefer the leasing channel $i$ over the sensing channel $j$, and this tendency increases as $\alpha_0$ decreases. If we view the channel unavailability $(1-\alpha_0)$ as the risk of choosing the sensing channel, then the PMC policy is a risk averse one. This is mainly due to the concavity of the rate function.
This preference order will also hold in the imperfect sensing case, in which case we will have $g_i>g_j$ and $\left(\log\left(\frac{g_i(t)}{\lambda(t)}\right)\right)^+\ge\alpha_0\left(\log\left(\frac{g_j(t)}{\lambda(t)}\right)\right)^+$.}

\subsection{Performance of the PMC Policy}
\label{sec:performance}
We can characterize the performance of the PMC Policy as follows: %The proof details are in Appendix~\ref{sec:proof_th2_a} and~\ref{sec:proof_th2_b} }.

%First, we can prove that the queue lengths are upper-bounded as the following inequalities.
\begin{theorem}
\label{th:performance}
For any positive value $V$, the PMC Policy has the following properties:
\begin{itemize}
    \item[(a)] The queue stability (\ref{eq:queue_stable}) and collision constraints  (\ref{eq:collision_con}) are satisfied. The queue length is  upper bounded by
\begin{eqnarray}
   Q(t)\le Q_{\max}\mdef Vq_{\max}\!+\!A_{\max}, \quad\forall\, t;
\label{eq:queue_bound}
\end{eqnarray}
and the virtual queue length is upper-bounded by
\begin{eqnarray}
   Z_i(t)\le Z_{\max}\mdef \kappa (V q_{\max}+A_{\max})+1 , \quad\forall\,i, t.
\label{eq:virtual_queue_bound}
\end{eqnarray}
where $$\kappa\!\mdef\! \frac{r_{\max}p_0(1-P_{fa}(C^{s0})))}{(1-p_0)P_{md}(C^{s0})}$$ and  $C^{s0}\mdef\max\limits_{C^s}\frac{p_0(1-P_{fa}(C^{s})))}{(1-p_0)P_{md}(C^{s})}$.
    \item[(b)] The average profit  $\overline{R_{PMC}}$ obtained by the PMC policy satisfies%\com{limit is over $t$?}
\begin{equation}
\label{eq:pmc_profit}
    %\lim\limits_{t\rightarrow \infty}}
\inf \overline{R_{PMC}}\ge \overline{R}^*-O(1/V),
\end{equation}
where $\overline{R}^*$ is the optimal value of the PM problem.
\end{itemize}
\end{theorem}

According to the Little's law, the average queuing delay is proportional to the queue length. Thus users experience bounded queuing delays under the PMC algorithm by (\ref{eq:queue_bound}).
By (\ref{eq:pmc_profit}), we find that the profit obtained by the PMC Policy can be made closer to the optimal profit by increasing $V$. However, as $V$ increases, the queuing delay also increases as shown in (\ref{eq:queue_bound}). The best choice of  $V$ depends on the desired trade-off between queuing delay and profit optimality.

A detailed proof of Theorem~\ref{th:performance} is provided in Appendices~\ref{sec:proof_th2_a} and \ref{sec:proof_th2_b}.

\subsection{Extension: More General Model of Primary Activities}
\label{sec:Markovian_model}
In the previous analysis, we have assumed that primary users' activities in each sensing channel follow a simple i.i.d. Bernoulli random process. {Next we will show that the PMC policy can be easily adapted to the} more general Markov chain model  of the primary users' activities shown in Fig.~\ref{fig:general_PU_model}. %: the sensing channel state $S_i(t)$ is only relates to the state in the last time $S_i(t-1)$, \ie the transition probability satisfies $Pr(S_i(t)|S_i(1)S_i(2)\dots,S_i(t-1) )=Pr(S_i(t)|S_i(t-1))$.
In this model, for any time $t$, $S_i(t)$ is unknown, but the history information $S_i(t-1)$ is known, and also the  transition probabilities $Pr(S_i(t)=s'|S_i(t-1)=s)\mdef p^i_{s\rightarrow s'}, s\in\{0,1\},s'\in\{0,1\}, i\in\mb^s_{\max}$ are known from long-time statistics.
%We consider a more general model for the activities of primary users in sensing channel $i\in\mb^s_{\max}$:
\begin{figure}[htb]
\centering
\includegraphics[scale=0.28]{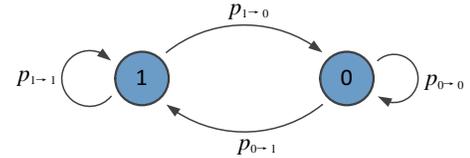}
\caption{Markov chain model of the PUs' activities}
\label{fig:general_PU_model}
\end{figure}

All previous analysis for PMC policy will still hold if we update two parameters $\omega_i(t)$ and $\alpha_i(t)$ as follows:
\begin{equation*}
%\label{eq:omega_g_model}
\omega_i(t)\mdef
  \begin{cases}
   \expectation\,[S_i(t)|W_i(t)=1,S_i(t-1)],&\text{if }i\in\mb^s(t) \\
   1,       &\text{if }i\in\mb^l(t)
  \end{cases}
\end{equation*}
where
\begin{align*}
 &\expectation\,[S_i(t)|W_i(t)=1,S_i(t-1)=s]\nonumber\\
&=\frac{p^i_{s\rightarrow 1}(1-P_{fa}(C^s(t)))}{p^i_{s\rightarrow 1}(1-P_{fa}(C^s(t)))+p_{s\rightarrow 0}P_{md}(C^s(t))};
\end{align*}
and
\begin{equation*}
%\label{eq:alpha_g_model}
    \alpha_i(t)=\begin{cases}
\expectation\,[S_i(t)W_i(t)|S_i(t-1)]& \text{if } i\in\mb^s(t)\\
1 & \text{if } i\in\mb^l(t)
\end{cases}
\end{equation*}
where
\begin{align*}
&\expectation\,[S_i(t)W_i(t)|S_i(t-1)=s]=p^i_{s\rightarrow 1}(1-P_{fa}(C^s(t))).
\end{align*}

There is no change in the revenue maximization part, and the cost minimization part still involves a combinatorial optimization problem. But the complexity of solving cost minimization problem becomes $\mathcal{O}(|\mathcal{C}|\times 2^{|\mb^s_{\max}|}|\mb^l_{\max}+1|)$, since we lose the structure information in sensing channels, \ie the threshold structure does not hold for sensing channels. In the worst case ($\mb_{\max}=\mb^s_{\max}$, $\mb^l_{\max}=\emptyset$ ), it comes back to $\mathcal{O}(|\mathcal{C}|\times 2^{|\mb_{\max}|})$, which is the complexity of the exhaustive search without considering the threshold structure.

Let us further consider a special Markov chain model where the transition probability for each sensing channel is the same, \ie  $Pr(S_i(t)=s'|S_i(t-1)=s)\mdef p_{s\rightarrow s'}, \forall i\in\mb^s_{\max}$. In this model, all sensing channels can be categorized into two types, channels being busy in the last slot (\ie $S_i(t-1)=0$), or channels being idle in the last slot (\ie $S_i(t-1)=1$). We can still show threshold structures for both types. Thus the complexity is reduced to $\mathcal{O}(|\mathcal{C}|\times|\mb_{\max}|^4\log(|\mb_{\max}|))$.

The above analysis shows that it is critical to exploit the problem structure to reduce the algorithm complexity.

\section{Heterogeneous Users}
\label{sec:MPMC_policy}
In Section~\ref{sec:problem_formulation}, we adopt the single queue analysis for homogeneous users who are assumed to be located nearby and have the same channel condition on  each channel. However, the single queue analysis no longer works for a more general scenario of heterogeneous users, where users can be located at different places, and have different channel conditions. In this section, we introduce the multi-queue model to deal with the heterogenous user scenario as shown in Fig.~\ref{fig:clusters}.

We divide the total coverage of the secondary base station into $\mj\mdef\{1,2,\dots,J\}$ disjoint small areas (illustrated as hexagons in Fig.~\ref{fig:clusters}) according to users' different channel experiences. Users in one of these small area are nearby homogeneous users. They share the same channel conditions, and  form a queue based on the FCFS discipline. We use $Q_j(t)$ to denote the queue length in area $j$. Since the queue and the corresponding area is one-to-one mapping, we also call the users in area~$j$ as queue~$j$ users.
\begin{figure}[t]
\centering
\includegraphics[scale=0.53]{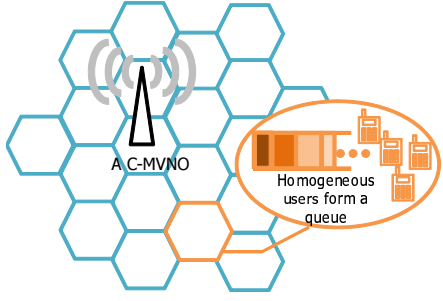}
\caption{Heterogeneous user model: Users in a hexagons are nearby homogeneous users, who have the same channel experience. Users in deferent hexagons can have different channel experience. }
\label{fig:clusters}
\end{figure}

For each queue $j\in\mj$, $h_{ij}(t)$ represents the users' channel gain  to channel $i\in\mb_{\max}$ at time $t$, which follows an i.i.d distribution over time. The indicator variable $T_{ij}(t)\in\{0,1\}$ denotes \textbf{the operator's channel assignment at time $t$}: $T_{ij}(t)=1$ if channel $i$ is allocated to queue~$j$, and $T_{ij}(t)=0$ otherwise. Meanwhile, the assignment $T_{ij}$ must satisfy
\begin{equation}
\label{eq:cluster_time_con}
    \sum_{j\in\mj}T_{ij}(t)\le 1, \forall t.
\end{equation}

The power allocation for queue~$j$ on channel $i$ is denoted by $P_{ij}(t)$. The total power allocation must satisfy
\begin{equation}
\label{eq:cluster_power_con}
    \sum_{j\in\mj}\sum_{i\in\mb_{\max}}P_{ij}(t)\le P_{\max}, \forall t,
\end{equation}
Thus the rate of queue~$j\in \mj$ can be calculated
\begin{equation}
r_j(t)=\sum_{i\in\mb(t)}I_i(t)T_{ij}(t)\log\left(1+\frac{h_{ij}(t)P_{ij}(t)}{T_{ij}(t)}\right)
\end{equation}
where $I_i(t)$ is the transmission result in channel $i$, following the same definition in Section~\ref{sec:power}.

For each queue~$j$, we follow the same demand model as in Section~\ref{sec:demand_model} for homogenous users. We assume the number of incoming users, the market state, and  the user's instantaneous demand are i.i.d among different queues, and denote them as $N_j(t)$, $M_j(t)$, and $A_j(t)$ respectively.
The market control variable and price for queue~$j$ are denoted as $O_j(t)$ and $q_j(t)$. The queuing dynamic for queue~$j$ is as follows:
\begin{equation}
\label{eq:micro_network_queue}
    Q_j(t+1)=\left(\!Q_j(t)-r_j(t)\right)^++ O_j(t)A_j(t).
\end{equation}
Thus the homogeneous user model in Section~\ref{sec:problem_formulation} can also be viewed as a special case of the heterogeneous user model, where $\mj$ is a singleton.

\subsection{Multi-queue Profit Maximization Control Policy}
\subsubsection{Revenue Maximization}
For any queue $j\in \mj$, we compute the optimal transmission price $q_j^{\ast}(t)$ by solving the following problem.
\begin{align}
\text{Maximize}\;\;& \left(q_j(t)-\frac{Q_j(t)}{V}\right)D\left(q_j(t),M_j(t)\right)\label{eq:cluster_pricing}\\
\text{Variables}\;\;& q_j(t)\ge 0\nonumber
\end{align}
If the maximum objective in (\ref{eq:cluster_pricing}) is positive, the operator sets transmission control variable $O_j^{\ast}(t)=1$ and accepts users' new file download requests at the price $q_j^\ast(t)$. Otherwise, the operator sets $O_j^{\ast}(t)=0$ and rejects any new requests.

\subsubsection{Cost Minimization}
\label{sec:sensing_resource}
%We solve the following optimization problem to determine sensing cost and resource allocation:
%\begin{align}
%\text{Minimize}& \sum_{i\in \mb(t)}\!\! C_i(t)- \sum_{j\in\mj}T_{ij}\omega_i(t)\log \left(1+\frac{h_{ij}(t)P_{ij}(t)}{T_{ij}(t)}\right)\label{eq:cluster_r_problem}\\
%\text{Subject to}&\quad (\ref{eq:sensing_bandwidth_con}), (\ref{eq:leasing_bandwidth_con}), (\ref{eq:cluster_power_con}),(\ref{eq:cluster_time_con}) \nonumber\\
%\text{Variables}&\quad C^s(t),\mb^s(t),\mb^l(t), P_{ij}(t)\ge 0, T_{ij}(t)\in\{0,1\}\nonumber
%\end{align}

We solve the following optimization problem to determine sensing cost and resource allocation:
\begin{align}
\text{Minimize}& \sum_{i\in \mb(t)}\!\! C_i(t)-\sum_{j\in\mj} \frac{Q_j(t)}{V}\expectation\,[r_j(t)]\label{eq:mul_r_problem}\\
\text{Subject to}&\quad (\ref{eq:sensing_bandwidth_con}), (\ref{eq:leasing_bandwidth_con}), (\ref{eq:cluster_power_con}),(\ref{eq:cluster_time_con})  \nonumber\\
\text{Variables}&\quad C^s(t),\mb^s(t),\mb^l(t), P_{ij}(t)\ge 0, T_{ij}(t)\in\{0,1\}\nonumber
\end{align}
where the cost $C_i(t)$ follows the same definition of (\ref{eq:virtual_cost}).
Similar to the homogenous model in Section\ref{sec:cost_min}, it is a two-stage decision problem. In the first stage, we determine the sensing technology, and choose the sensing channels and the leasing channels. Then in the second stage, we determine the channel assignment and power allocation based on the sensing results. We use backward induction to solve this problem (\ref{eq:mul_r_problem}). To simplify the notation, we will ignore the time index in the following analysis.

We first analyze the channel assignment and power allocation in the second stage.
 In this stage, since the sensing results $W_i$, the channels $\mb^s$, and the sensing technology $C^s$  are determined, the second stage problem of  (\ref{eq:r_problem}) is shown as follows:
\begin{align}
\text{Maximize}& \sum_{j\in\mj} \sum_{i\in \mb}Q_j\omega_iT_{ij}\log \left(1+\frac{h_{ij}P_{ij}}{T_{ij}}\right)\label{eq:mul_2nd_problem}\\
\text{Subject to}&\quad (\ref{eq:cluster_power_con}),(\ref{eq:cluster_time_con})\nonumber\\
\text{Variables}&\quad P_{ij}\ge 0, T_{ij}\in\{0,1\}\nonumber
\end{align}
where $\omega_i$ follows the same definition of (\ref{eq:omega}).

Compared to the power allocation problem in  (\ref{eq:power_allocation_problem}), the binary channel assignment variables $T_{ij}$'s make the second stage problem in (\ref{eq:mul_2nd_problem}) much more complex.

We first solve problem (\ref{eq:mul_2nd_problem}) assuming fixed $T_{ij}$'s, in which case the power allocation problem is a convex optimization problem. Following the same method of solving power allocation for homogeneous users as in (\ref{eq:power_allocation}), we have
\begin{equation}
\label{eq:mul_power_allocation}
P_{ij}=T_{ij}Q_j\omega_i\left(\frac{1}{\lambda}-\frac{1}{Q_j\omega_ih_{ij}}\right)^+,\; i\in\mb.
\end{equation}
where $\lambda$ is the Lagrange multiplier of the total power constraint (\ref{eq:cluster_power_con}), which satisfies
\begin{equation}
\label{eq:mul_dual_sol}
\lambda= \frac{\sum_{i\in \mb_{p}}\omega_iQ_jT_{ij}}{P_{\max}+\sum_{i\in\mb_{p}}\frac{T_{ij}}{h_i(t)}},
\end{equation}
where $\mb_{p}\mdef\{i\in\mb:P_{ij}>0\}$.
When $T_{ij}$ is known, we can design a simple search algorithm similar to Alg.~\ref{alg:power_allocation} to determine the optimal value of $\lambda$.

We then substitute this result in (\ref{eq:mul_2nd_problem}), and further maximize the objective over $T_{ij}$'s. Then we have
\begin{align}
\text{Maximize}& \;\sum_{i\in \mb}\omega_i\sum_{j\in\mj} Q_jT_{ij}\left(\log\left(\!\frac{\omega_iQ_ih_i}{\lambda}\!\right)\right)^+\label{eq:mul_2nd_problem_T}\\
\text{Subject to}&\;\sum_{j\in\mj}T_{ij}\le 1\nonumber\\
\text{Variables}&\;\; T_{ij}\in\{0,1\}\nonumber
\end{align}

For each channel $i\in\mb$, let us denote the set
\begin{equation}
\label{eq:cluster_time_set}
\mj_i^*\!\mdef\!\arg\max_{j\in\mj}Q_j\left(\!\log\left(\!\frac{\omega_iQ_jh_{ij}}{\lambda}\!\right)\!\!\right)^+
\end{equation}
Here the solution $\mj_i^*$ is a set of indices of the chosen queues.
If $\mj_i^*$ is a singleton, then we denote its unique element as $j_i^*$. If $\mj^*$ is not singleton, since all elements in $\mj^*$ lead to the same value in objective, then we can randomly pick one of the its element and denote it as $j_i^*$.

Since the objective in problem (\ref{eq:mul_2nd_problem_T}) is linear in $T_{ij}$, it is easy to see the optimal solution is
\begin{equation}
\label{eq:cluster_time_sol}
T_{ij}=\begin{cases}
1, &\text{if } j=j_i^*,i\in\mb\\
0, &\text{otherwise}.
	 \end{cases}
\end{equation}

Now let us consider how to calculate the value of $\lambda$ and $j^*_i$. By (\ref{eq:mul_dual_sol}) and (\ref{eq:cluster_time_set}), we find that they are actually coupled together. To determine $\lambda$ in (\ref{eq:mul_dual_sol}), we need to know $T_{ij}$ (or equivalently $\mj_i^*$, $j^*_i$), \ie which queue is chosen for which channel. But determining $\mj_i^*$ in (\ref{eq:cluster_time_set}) requires the value of  $\lambda$.
One way to solve this problem is to enumerate every possible channel assignment combinations to find the solutions satisfying both (\ref{eq:mul_dual_sol}) and (\ref{eq:cluster_time_set}). Since each channel $i\in\mb_{\max}$ can be assigned to one of $J$ queues, there are a total of $|\mb_{\max}|^J$ channel assignment combinations. When the $J$ or $|\mb_{\max}|$ is large, the complexity can be very high. However, we can reduce the search complexity by exploring the special structure of the problem.
\begin{property}
For each channel allocated positive power $i\in\mb_p$, we have
\begin{equation}
\label{eq:channel_selection}
    \mj_i^*=\arg \max\, \{Q_j\,|\,h_{ij}>\frac{\lambda}{Q_j\omega_i}, j\in\mj\}.
\end{equation}
\end{property}

This property comes from (\ref{eq:cluster_time_set}). This means that for a particular channel $i$, if the channel gain for the longest queue $Q_j$ is good enough (\ie $h_{ij}>\frac{\lambda}{Q_j \omega_i}$), we should assign channel $i$ to queue  with the longest queue length. {If $\mj_i^*$ is a singleton, then we denote its unique element as $j^*_i$. Otherwise, we denote $\hat{\mj}^{*}_i\mdef \arg \max\, \{h_{ij}\,|\, j\in\mj_i^*\}$. In this case, there are multiple channels with the same channel gain and the same queue length, and we can randomly pick one and denote it as $j^*_i$.}
Thus we can search $\lambda$ and $\mj_i^*$ (and also $j_i^*$) by a simple greedy algorithm as follows. First, for all channels, we assume $\mj^*_i=\arg \max_{j\in\mj}\,Q_j$, and calculate the value of $\lambda$ by the waterfilling algorithm (as the procedure ``Waterfilling'' in Alg.~\ref{alg:channel_assignment}).
For each unchosen channel, \ie the channel $i$ with $\omega_iQ_{j^*_i}h_{ij^*_i}\le\lambda$, we check whether $\omega_iQ_{j}h_{ij}>\lambda$ can be satisfied when another queue is chosen instead of $j^*_i$. If there is some set $\tilde{J}_i$ of queues satisfying $\omega_iQ_{j}h_{ij}>\lambda$, we replace ${\mathcal{J}}^*_i$ with the one with the longest queue length in this set, \ie $\arg \max_{j\in\tilde{J}_i}\,Q_j$. %\com{I revise the notations to be consistent with (51), but it is still a bit confusing. In general, we should reserve mathcal letter for set, and low case letter for index. However, (51) use mathcal letter for index. Let's fix it. Also check the notations for the rest of this paragraph.}
We repeat the process iteratively until we find the $\lambda$ and $j^*_i$ that satisfies both (\ref{eq:mul_dual_sol}) and (\ref{eq:cluster_time_set}).
The pseudo code is given in Alg.~\ref{alg:channel_assignment}. To simplify the expression of Alg.~\ref{alg:channel_assignment}, with a little abuse of notations, we denote $h_i=h_{ij^*_i}$, $Q_i=Q_{j^*_i}$, and  $\Lambda(m)\mdef \frac{\sum_{i=1 }^m\omega_iQ_i}{P_{\max}+\sum_{i=1}^m\frac{1}{h_i}}.$

\begin{algorithm}[hbt]
\caption{Channel Assignment}
\label{alg:channel_assignment}
\begin{algorithmic}[1]
\State $\mj^*_i\gets\arg \max_{j\in\mj}\,Q_j$
\Statex
\Procedure {Waterfilling} {$h_i(t),Q_i(t)$}
\State Rearrange the channel indices $i\in\mb_{\max}$ in the decreasing order of $\omega_i(t)Q_{i}(t)h_{i}(t)$
\State $m\gets |\mb(t)|$, $\lambda\Leftarrow\Lambda(m)$ %\% Initialization
\While{$\lambda\ge h_m(t)\omega_m(t)$}
	\State $m\gets m-1$		
	\State $\lambda\Leftarrow\Lambda(m)$
\EndWhile
\EndProcedure
\Statex
\While{$\omega_i(t)h_{ij}(t)Q_j(t)>\lambda,\forall j, \,\forall i>m, $}
	\State $\mj^*_i\gets\arg \max\,\{Q_j|\omega_i(t)h_{ij}(t)Q_j(t)>\lambda\}$		
	\State {\bf invoke} procedure {Waterfilling}( {$h_i(t),Q_i(t)$})
\EndWhile
%\EndProcedure
\end{algorithmic}
\end{algorithm}

The complexity of Alg.~\ref{alg:channel_assignment} is $\mathcal{O}(|\mb_{\max}|^3\log(|\mb_{\max}|)$, since the while loop runs no more than $|\mb_{\max}|$ times in the worst case, and the complexity of waterfilling part is $\mathcal{O}(|\mb_{\max}|^2\log(|\mb_{\max}|)$ (the same as the waterfilling power allocation algorithm in Section~\ref{sec:2nd_stage}).
Compared to the exhaustive search, the complexity of solving the channel assignment is greatly reduced.
%Since there are $|\mb_{\max}|-|\mb_p|$ channels to be checked and updated, the while loop in Alg.~\ref{alg:channel_ass} runs $|\mb_{\max}|-|\mb_p|$ times.

With the solution of channel assignment, we can update the power allocation solution of (\ref{eq:mul_r_problem}) as
 $$P_{ij}^*=\begin{cases}
\omega_iQ_j\left(\frac{1}{\lambda}-\frac{1}{\omega_iQ_jh_{i}}\right)^+, &\text{if } j=j^*, i\in\mb,\\
0 & \text{otherwise},
\end{cases}
$$
where the value of $\lambda$, $h_i$ and $Q_i$ are calculated by Alg.~\ref{alg:channel_assignment}.

After solving the second stage problem, we move to the first stage. Following the channel assignment in (\ref{eq:cluster_time_sol}), we find that the first stage problem for the heterogeneous users is the same with the one of homogeneous users problem in (\ref{eq:r_problem_2}). We can simply run the same Alg.~\ref{alg:sensing_cost_search} to determine sensing technology $C^{s*}(t)$,  sensing channel set $\mb^{s*}(t)$ and leasing channel set $\mb^{l*}(t)$.

\subsection{The Performance of M-PMC Policy}

Next we show the performance bounds of the M-PMC Policy. The proof method is similar to that of Theorem~\ref{th:performance}, and the details are omitted due to space limit.

%First, we can prove that the queue lengths are upper-bounded as the following inequalities.
\begin{theorem}
\label{th:performanceM}
%Suppose that system parameters $\phi(t)$'s are i.i.d. over time slots.
For any positive value $V$, the M-PMC Policy has the following properties:
%with probabilities p(?), the problem (4.31)-(4.35) is feasible, and that E{L((0))} < 8. Fix a value
%C = 0. If we use a C-additive approximation of the algorithm every slot t , then:

\begin{itemize}
    \item[(a)] The queue stability (\ref{eq:queue_stable}) and collision constraints  (\ref{eq:collision_con}) are satisfied. The queue length is  upper bounded by%\com{I removed the finite queue length result (redundant with the following equation) and the virtual queue (an intermediate result for proof).}
\begin{eqnarray}
   Q_j(t)\le Q^{\max}\mdef Vq^{max}\!+\!A^{max}, \quad\forall\, t,
\label{eq:cluster_queue_bound}
\end{eqnarray}
and the virtual queues are bounded by
\begin{eqnarray}
   Z_i(t)\le Z^{\max}\mdef \kappa (V q^{\max}+A^{\max})+1 , \quad\forall\,i, t.
\label{eq:cluster_virtual_queue_bound}
\end{eqnarray}
where $\kappa\mdef r_{\max}p_0\frac{p_0(1-P_{fa}(C^s(t)))+(1-p_0)P_{md}(C^s(t))}{(1-p_0)}$, and $C^{s0}$ denotes the highest sensing cost.
    \item[(b)] The average profit  $\overline{R_{M-PMC}}$ obtained by the M-PMC policy satisfies
\begin{equation}
\label{eq:utility_max_proof}
    %\lim\limits_{t\rightarrow \infty}}
\inf \overline{R_{M-PMC}}\ge \overline{R}^*-O(1/V),
\end{equation}
where $\overline{R}^*$ is the optimal value of the multi-queue PM problem.
\end{itemize}
\end{theorem}

\section{Simulation}
\label{sec:simulation}
In this section we provide simulation results for PMC and M-PMC policies.

We conduct simulations with the following parameters. The number of incoming users in each slot satisfies a Poisson distribution with a rate $D(q(t),M(t))=\frac{1}{M(t)}(q(t)-5)^2$. The market state $M(t)$ satisfies  Bernoulli distribution, $M(t)=1$ with probability 0.5, and $M(t)=2$ with probability 0.5. The file length of each user satisfies the i.i.d. (discrete) uniform distribution between 1 and 10. There are 32 channels in total. 20 of them belong to the sensing band $\mb^s_{\max}$, and the rest 12 channels belong to the leasing band $\mb^l_{\max}$. The primary collision probability tolerant levels are set as $\eta_i=0.001$ for sensing channel $i=1,2,\dots,10$, and $\eta_i=0.005$ for sensing channel $i=11,12,\dots,20$. The channel gain $h_i$ of each channel satisfies i.i.d. (continuous) Rayleigh distribution with parameter $\sigma=4.5$ %(such that virtual SNRs of most channels lie in range of 10 -- 30dB).
The total power constraint of the base station is $P^{\max}=8$.
 There are 3 different sensing technologies with costs $\mc^s=\{0\, (\text{not sensing at all}), 0.1,0.5\}$. The corresponding false alarm  probabilities are $P_{fa}=\{0.5,0.1,0.008\}$, and the missed detection probabilities are $P_{md}=\{0.5,0.08,0.005\}$.
We assume the idle time probability of sensing band $p_0$ is 0.6, and the control parameter $V\in\{5, 10, 50, 100, 200\}$.

Figure~\ref{fig:collision} shows a collision situation of all sensing channels with the control parameter $V=100$. We find that primary users' collision tolerant bound (\ref{eq:collision_con}) is satisfied as time increases. We also find that we obtain similar curves for the collision probabilities with other values of control parameter $V$.

\begin{figure}[htb]
\centering
\includegraphics[scale=0.37]{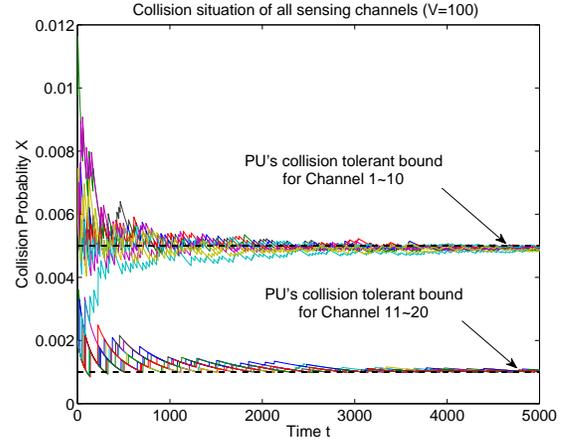}
\caption{A collision situation of all sensing channels with $V=100$}
\label{fig:collision}
\end{figure}

Figure.~\ref{fig:queue_length} (a) shows that the average queue length grows linearly in $V$,  and is always less than the worst case bound $Vq^{max}+A^{max}$.
Figure.~\ref{fig:queue_length} (b) shows that
the average profit achieved by PMC policy converges quickly as $V$ grows, and  is close to the maximum profit when $V \geq 100$.

\begin{figure}[htb]
\centering
\includegraphics[scale=0.36]{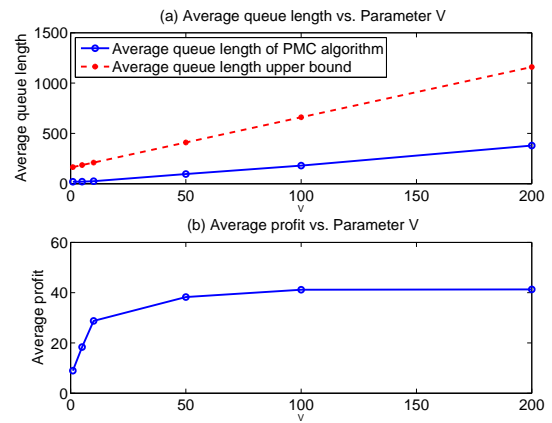}
\caption{(a) Average queue length vs. Parameter $V$, (b) Average profit vs.  Parameter $V$}
\label{fig:queue_length}
\end{figure}

\begin{figure}[htb]
\centering
\includegraphics[scale=0.38]{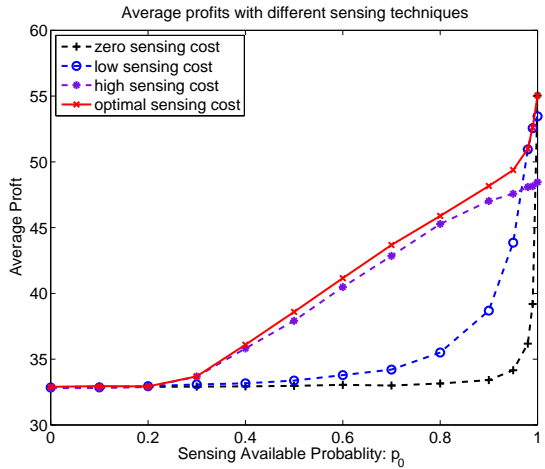}
\caption{Average Profit with different sensing technologies}
\label{fig:sensing_cost}
\end{figure}

We further vary the idle time probability of sensing band $p_0$ from 0 to 1. In Fig.~\ref{fig:sensing_cost}, we show the average profit with different sensing available probabilities $p_0\in [0,1]$ and different fixed sensing technologies. The black curve is with zero sensing cost, where $P_{fa}=P_{md}=0.5$, which means the operator does not perform sensing and takes random guesses of primary users' activities in sensing channels. The blue curve is with the low sensing cost $0.1$, where $P_{fa}=0.1$ and $P_{md}=0.08$. The purple curve is with the high sensing cost $0.5$,  where $P_{fa}=0.008$ and $P_{md}=0.005$. The red curve is the PMC policy, which adaptively choose sensing cost from the above three (sensing cost $C_s=0,0.1,0.5$). When the sensing available probability is small (e.g., $p_{0}\in [0,0.2]$), all strategies tend to only choose the leasing channels, thus all curves obtain similar profits. When the sensing available probability $p_{0}$ further increases, the advantage of exploring sensing channels becomes more significant. The performances for strategies using the zero and low sensing cost are not good. The reason is that their detection accuracy is not good enough. To achieve the primary users' collision bounds, these strategies choose sensing channels less often, and replace with more expensive leasing channels. When the sensing available probability is high and close to 1,  sensing seems unnecessary. Therefore, the performance increases as sensing cost decreases, where the zero cost is the best and the high cost is the worst. The PMC policy (the red curve) adaptively chooses the sensing cost, {\ie when $p_{0}$ is medium, it utilizes the high sensing cost strategy in most of time slots; as $p_0$ keeps increasing, it gradually changes to utilize the low cost and zero cost strategies more frequently; and when $p_0$ goes to 1, it utilizes the zero sensing cost strategy in most of time slots. It has the best performance, since it can take advantage of different sensing technologies for different sensing available probabilities.}
%\com{same question as before: why the optimal case is strictly better than the high sensing cost in the middle? Should't they overlap with each other?}\res{PMC policy is a smart combination of three strategies. In the middle range of $p_0$, it chooses high cost strategy very often, since high cost strategy has the best performance among the three. But it sometimes choose the other two as well. The fact that the performance of PMC policy is even better than high cost strategy implies high cost strategy is a little bit too conservative. Sometimes choosing more risky strategies (\eg zero or low costs) opportunistically can further improve the performance. We also find that as $p_0$ increasing, PMC policy will adaptively increase the frequency of choosing the risky strategies.}

For the  M-PMC policy, we conduct simulations for a simple two-queue system. For queue-1, the channel gain $h_{i1}$ of each channel satisfies i.i.d. (continuous) Rayleigh distributions with parameter $\sigma=4.5$. For queue-2, the channel gain $h_{i2}$ of each channel satisfies i.i.d. (continuous) Rayleigh distributions with parameter $\sigma=5.5$. This is because  queue-2 users are closer to the operator's base station than queue-1 users.

Figure~\ref{fig:multi_Q} (a) shows that the average transmission rate obtained by queue-2 users is higher than that of the queue-1 users. This is because queue-2 users usually have better channel conditions, and M-PMC policy prefers to allocate more powers to better channels to improve the transmission rate. Figure~\ref{fig:multi_Q} (b) shows that  the revenues obtained by the operator from users of two queues are almost the same when all queues are stable. It is an interesting observation. We can understand it in this way: when two queues make different revenue, to maximize the profit (also the revenue), the operator will allocate more transmission rates to the queue with a higher revenue. Thus the length of the queue with a higher transmission rate will be shortened, and the negative queuing effect in revenue maximization problem will be soon diluted. It further leads to a decreasing price and a decreasing revenue. In contrast, the length of the queue with a lower transmission rate increases, which results in an increasing price and an increasing revenue. Therefore, when all queues are stable finally, the average revenue generated by each queue is the same.

%Since the average revenue equals to the average rate times the average price, we can see that the average price for queue 1 users is higher than that of queue 2 users. That is because the operator usually takes more efforts (power) to transmit a packet for queue 2 users, to maximize its profit, it must charge the queue 2 users a higher price. %Therefore in the real system where some fairness condition should be enforced, we can realize this fairness condition by adding corresponding weights in the revenue functions of users from different queues.

\begin{figure}[htb]
\centering
\includegraphics[scale=0.39]{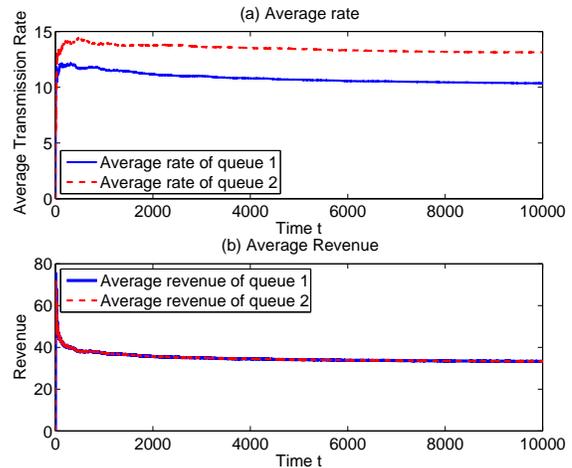}
\caption{(a) Average transmission rates of a two-queue M-PMC policy, (b) Average revenues  of a two-queue M-PMC policy }
\label{fig:multi_Q}
\end{figure}
%\begin{figure}[t]
%\centering
%\includegraphics[scale=0.45]{profit}
%\caption{Average profit vs. parameter V.}
%\label{fig:profit}
%\end{figure}

\section{Conclusion}
\label{sec:conclusion}
In this paper, we study the profit maximization problem of a cognitive mobile virtual network operator in a dynamic network environment. We propose low-complexity PMC and M-PMC policies which perform both revenue maximization with pricing and market control, and cost minimization with proper resource investment and allocation. We show that these policies can achieve arbitrarily close to the optimal profit, and have flexible trade-offs between profit optimality and queuing delay.

We also find several interesting features in these close-to-optimal policies. In revenue maximization,  the dynamic pricing strategy performs the functionality of congestion control to users' demands, \ie the longer the queue length of demands, the higher price the operator should charge. In cost minimization, the operator is risk averse towards spectrum investment, and prefers stable leasing spectrum to unstable sensing spectrum with the same channel condition and the same availability-price-ratio. %\com{last term no defined in the paper?}.

%There are several interesting future directions to extend our work. Our solution in this paper is based on the elastic traffic model.
In this paper, we only looked at the issue of elastic traffic. It  would be worthwhile to incorporate inelastic traffic, which usually has strict constraints on transmission rates and delays. Typical examples include real-time multimedia applications, \eg audio streaming, Video on Demand (VoD), and Voice over IP (VoIP). In the most general case, we can consider a hybrid system with both elastic and inelastic traffic, which is more realistic and practical. In addition, as mentioned in Section~\ref{sec:literature}, the literature about competition in cognitive radio networks mainly focus on the static network scenario. It is also interesting to extend our dynamic model to incorporate competition among several network operators.

%\bibliographystyle{IEEEtran}
%\bibliography{CVNO}

\begin{IEEEbiography}[{\includegraphics[width=1in,height=1.25in,clip,keepaspectratio]{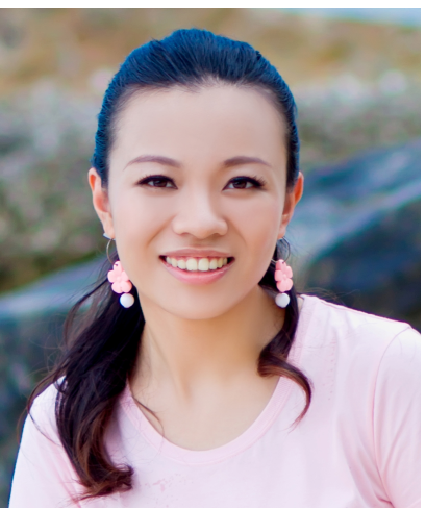}}]{Shuqin Li } received Ph.D. in Information Engineering from The Chinese University of Hong Kong in 2012. She now works as a research scientist in Research and Innovation, Alcatel-Lucent Shanghai Bell Co., Ltd. Her research interests include resource allocation, pricing and revenue management in communication networks, applied game theory, contract theory and incentive mechanism design in network economics, network coding and stochastic network optimization.
\end{IEEEbiography}

\begin{IEEEbiography}[{\includegraphics[width=1in,height=1.25in,clip,keepaspectratio]{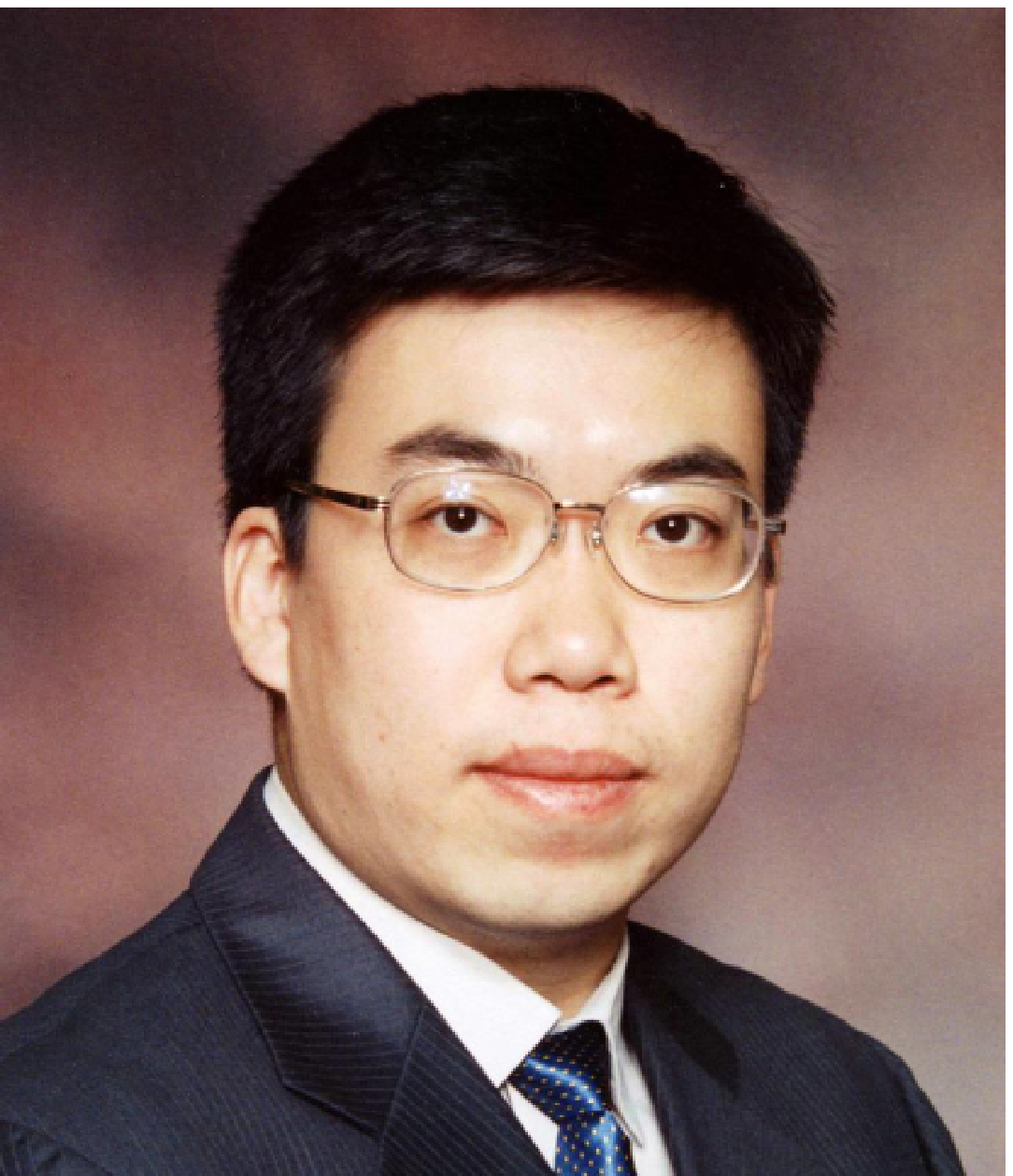}}]{Jianwei Huang}
is an Assistant Professor in the Department of
Information Engineering at the Chinese University of Hong Kong. He
received Ph.D. in Electrical and
Computer Engineering from Northwestern University in 2005. Dr. Huang currently
leads the Network Communications and Economics Lab
(ncel.ie.cuhk.edu.hk). He is the co-recipient of five best paper awards, including the 2011  IEEE Marconi Prize Paper Award in Wireless Communications .

Dr. Huang currently serves as Editor of \emph{IEEE Transactions on Wireless
Communications}, Editor of \emph{IEEE Journal on Selected
Areas in Communication - Cognitive Radio Series}, and Chair of IEEE ComSoc Multimedia Communications Technical Committee.
%He has also been
%the TPC Co-Chair of IEEE WiOpt 2012, GameNets 2009,  IEEE GlOBECOM Wireless Communications Symposium 2010, and IWCMC  Mobile Computing Symposium 2010.
\end{IEEEbiography}

\begin{IEEEbiography}[{\includegraphics[width=1in,height=1.25in,clip,keepaspectratio]{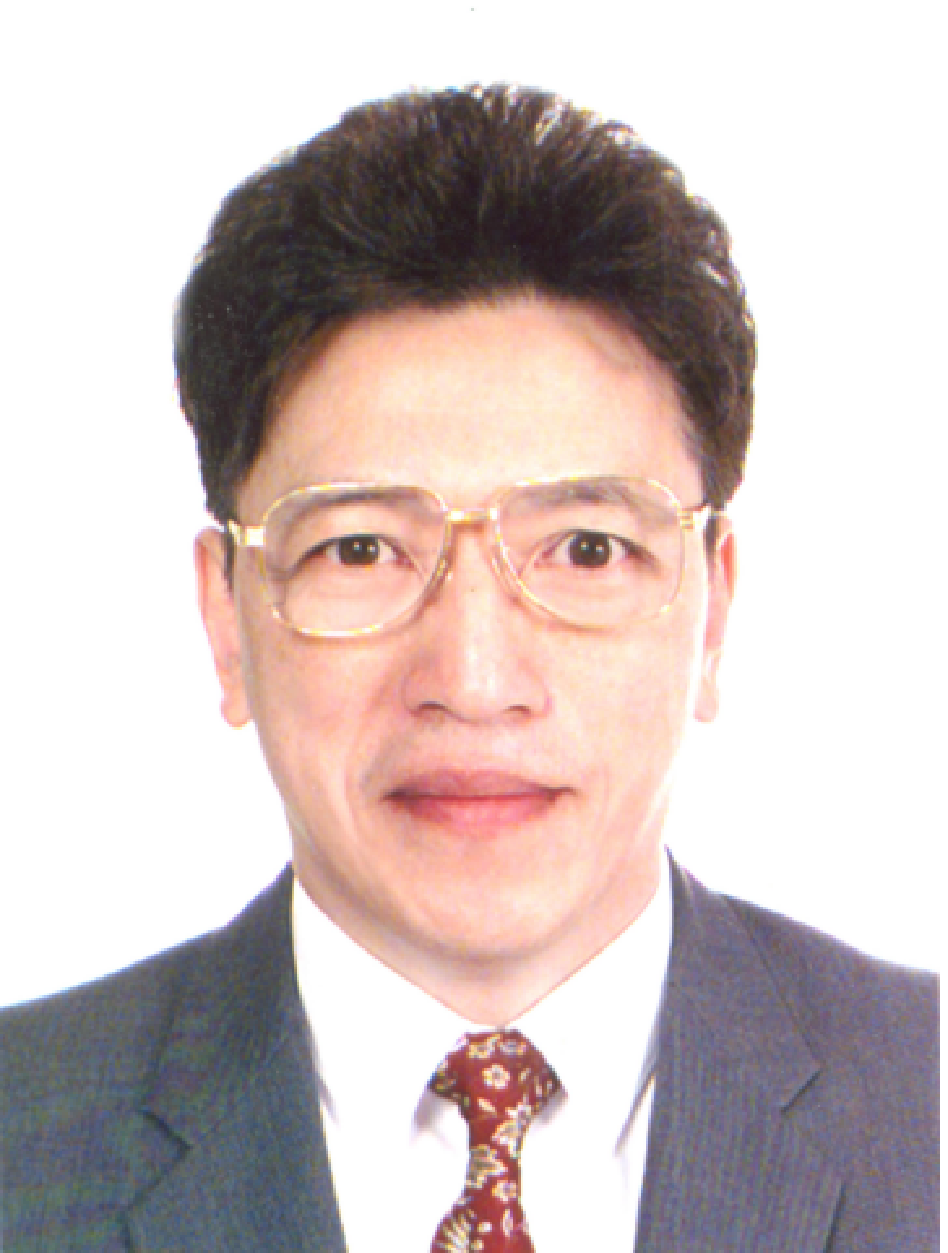}}]{Shuo-Yen Robert Li } S.-Y. Robert Li (Bob Li) received the PhD degree from UC Berkeley in 1974. He taught applied math at MIT in 1974-76 and math/statistics/CS at U. Illinois, Chicago in 1976-79. After working at Bell Labs/Bellcore for a decade, he joined CUHK as a chair professor in 1989. He now also serves as Co-director of Institute of Network Coding as well as honorary professor in a few universities. Bob initiated algebraic switching theory and is a cofounder of network coding theory. His "martingale of patterns" (1980) is widely applied to genetics.

\end{IEEEbiography}

\newpage
\appendices

\section{Proof for Theorem~\ref{th:performance} (a)}
\label{sec:proof_th2_a}
We first prove (\ref{eq:queue_bound}) by induction.

It is easy to see that at slot $t=0$, no packets are in the network and $Q(0)=0$, thus the queue length bound (\ref{eq:queue_bound}) obviously holds.
Now suppose (\ref{eq:queue_bound}) holds for time $t$. We consider the queue length bound in slot $t+1$ in the following cases:
\begin{itemize}
    \item \emph{Case 1:} $Q(t)\le Vq_{\max}$, then clearly $Q(t+1)\le Vq_{\max}+A_{\max}$.
	 \item \emph{Case 2:} $Q(t)> Vq_{\max}$, \ie the objective value of the revenue maximization (\ref{eq:pricing}) is negative. Therefor, according to the optimization solution, the operator will set $O(t)=0$, and do not accept any new request, $A(t)=0$. Therefore, $Q(t+1)\le Q(t)\le Vq_{\max}+A_{\max}$.
\end{itemize}

Likewise, we can also prove the virtual queue bound (\ref{eq:virtual_queue_bound}) by induction.  Suppose that the inequality (\ref{eq:virtual_queue_bound}) holds for time $t$, and consider the following two cases.

\begin{itemize}
\item \emph{Case 1:} $Z_i(t)\le Z_{\max}-1$, clearly the virtual queue length bound (\ref{eq:virtual_queue_bound}) also holds at slot $t+1$.

\item \emph{Case 2:} $Z_i(t)> Z_{\max}-1$, %=\max\limits_{C^s} \frac{f_i(C^s(t))}{(1-p_0)\beta(C^s(t))}r^{\max} Q^{\max}$,
\ie
\begin{equation}
\label{eq:virtual_queue_proof}
Z_i(t)\ge \max\limits_{C^s} \frac{r_{\max} Q_{\max}\alpha_i(t)}{(1-p_0)P_{md}(C^{s}(t))}
\end{equation}
Since $r_{\max}\ge \sum_{i\in\mb_{\max}}r_i(t)=\sum_{i\in\mb_{\max}}\log\left(\!\frac{\omega_i(t)h_i(t)}{\lambda(t)}\!\right)$,  then inequality (\ref{eq:virtual_queue_proof}) implies
\begin{align}
C_i(t)\!=\!C^s(t)\!+\!Z_i(t)\frac{1\!-\!p_0}{V}\!>\!\alpha_i(t)\log\left(\!\frac{\omega_i(t)h_i(t)}{\lambda(t)}\!\right).
\end{align}

By (\ref{eq:set_selection_sol}) in the PMC policy, channel $i$ will not be chosen for sensing and transmission, thus there will be no collision in this channel, \ie $X_i(t+1)=0$. Then by (\ref{eq:virtual_queue}) we have $Z_i(t+1)\le Z_i(t)$ and
the virtual queue length bound (\ref{eq:virtual_queue_bound}) also holds at $t+1$.
\end{itemize}

\section{Proof for Theorem~\ref{th:performance} (b)}
\label{sec:proof_th2_b}

We first construct a {\bf stationary} randomized policy that can achieve the optimal solution of the PM problem. Let us consider a special class of stationary randomized policies, called the \emph{$\phi$-only policy}, which makes the decision $\gamma(t)$ in slot $t$ only depending on the observation of system parameter $\phi(t)$. The stationary distribution for the observable parameter $\phi(t)$ is denoted as $\{\Pi_{\phi}, \phi\in\Phi\}$. (Recall that we define $\phi(t)\mdef(M(t),{\boldsymbol{h}}(t),{C^l}(t))$ and $\gamma(t)\mdef(O(t),q(t),{C^s}(t), \mb^s(t), \mb^l(t), \boldsymbol{P}(t))$. Note that the value of $\phi(t)$ can be chosen only from a finite set $\Phi$.) In the $\phi$-only policy, when the operator observes $\phi(t)=\phi$, it chooses $\gamma(t)$ from the countable collection of ${\Gamma_\phi (t)}=\{\gamma_{\phi}^1,\gamma_{\phi}^2,\dots \}$ with probabilities $\{\rho^1_{\phi},\rho^2_{\phi}\dots\}$, where $\sum_{u=1}^\infty \rho^u_{\phi}=1$. Note that the decision is independent of time $t$, and thus is stationary. We have the following fact:

There exists a stationary $\phi$-only policy that achieves the optimal profit of the PM problem while satisfying stability condition \eqref{eq:queue_stable} and collision upper-bound requirement \eqref{eq:collision_con}, which is the solution of the following optimization problem:
%
%Specifically, to construct this optimal $\phi$-only policy, we only need determine the optimal values of the probabilities $\{\rho^{1*}_{\phi},\rho^{2*}_{\phi}\dots\}$, which can be obtained by directly solving the PM problem based on $\phi$-only policy as follows.
\begin{align*}
\overline{R}^*=\text{Maximize}\;& \sum_{\phi\in\Phi}\Pi_{\phi}\sum_{u=1}^{\infty}R(\gamma_{\phi}^u;\phi)\rho^u_{\phi}\nonumber\\
\text{Subject to}\;& \sum_{\phi\in\Phi}\Pi_{\phi}\sum_{u=1}^{\infty}r(\gamma_{\phi}^u;\phi)\rho^u_{\phi}\le\sum_{\phi\in\Phi}\Pi_{\phi}\sum_{u=1}^{\infty}D(\gamma_{\phi}^u;\phi)\rho^u_{\phi}\\
&   \sum_{\phi\in\Phi}\Pi_{\phi}\sum_{u=1}^{\infty}X_i(\gamma_{\phi}^u;\phi)\rho^u_{\phi}\le \eta_i,\;\; i\in\mb^s_{max}
\end{align*}
The above fact is a special case of Theorem 4.5 in \cite{neely2010stochastic}. The proof is omitted for brevity.

Recall that the PMC policy is derived by minimizing the right hand side of the following inequality
\begin{align}
\label{eq:drift_penalty_expand2}
&\Delta(\boldsymbol{\Theta}(t))-V\expectation\,[R_{PMC}(t)|\boldsymbol{\Theta}(t)]\le D - V\expectation\,[R(t)|\boldsymbol{\Theta}(t)] \nonumber\\
&\quad\quad+Q(t)\expectation\,\left[O(t)A(t)-r(t)|\boldsymbol{\Theta}(t)\right] \nonumber\\
&\quad\quad+\sum_{i\in\mb^s_{\max}}Z_i(t)\expectation\,\left[X_i(t)-\eta_i|\boldsymbol{\Theta}(t)\right].
\end{align}
In other words, given the current queue backlogs for each slot $t$, the PMC policy minimizes the right hand side of \eqref{eq:drift_penalty_expand2} over all alternative feasible policies that could be implemented, including the optimal stationary $\phi$-only policy. Therefore, by plugging the optimal stationary $\phi$-only policy in the right hand side of \eqref{eq:drift_penalty_expand2}, we have
\begin{equation}
\label{eq:drift}
    \Delta(\boldsymbol{\Theta}(t))-V\expectation\,[R_{PMC}(t)|\boldsymbol{\Theta}(t)]\le D - V\overline{R}^*.
\end{equation}

Now we use the following lemma to obtain the performance bound in Theorem~\ref{th:performance}(b). %by applying Theorem~\ref{th:Lyapunov}.

%We first prove the following theorem.
\begin{lemma}
\label{th:Lyapunov}
 (Lyapunov Optimization) Suppose there are finite constants $V>0$, $D>0$,  such that for all time slots $t\in\{0,1,2,\dots\}$ and all possible values of $\boldsymbol{\Theta}(t)$, we have
\begin{equation}\label{eq:Lyapinov_opt}
  \Delta(\boldsymbol{\Theta}(t))-V\expectation\,[R(t)|\boldsymbol{\Theta}(t)] \le D- V\overline{R}^*.
\end{equation}
Then we have the following result
\begin{equation}\label{eq:optimality}
    {\underset{t\rightarrow\infty}{\lim\sup}\frac{1}{t}\sum_{\tau=0}^{t-1}\expectation\,[R(t)]}\ge \overline{R}^* - \frac{D}{V}.
\end{equation}
\end{lemma}
The above lemma is a special case of Theorem~4.2 in \cite{neely2010stochastic}. The proof is omitted for brevity.

Note that the inequality \eqref{eq:drift} is exact the condition \eqref{eq:Lyapinov_opt} in Lemma~\ref{th:Lyapunov}, thus the performance bound in Theorem~\ref{th:performance}(b) immediately follows.

\section{Pseudo code of Algorithm~\ref{alg:water-filling}}
\label{sec:search_threshold}
\begin{algorithm}[hbt]
\caption{Search Threshold $\Upsilon^l$ (or $\Upsilon^s$) for a given $C^s(t)$}
\label{alg:water-filling}
\begin{algorithmic}[1]
\Procedure{SearchingThreshold} {$C^s$}
\State Rearrange the channel indecent $i\in\mb^l_{\max}$ (or $i\in\mb^s_{\max}$) as the decreasing order of $g_i(t)$.
\State  $m\gets |\mb^l_{\max}|$ (or  $m\gets |\mb^s_{\max}|$), $\lambda\Leftarrow\Lambda(m)$ %\% Initialization
\While{$\lambda\ge g_m(t)$}
	\If {$m>1$}
	\State $m\gets m-1$		
	\State $\lambda\Leftarrow\Lambda(m)$
	\Else \State   \bf{break}
	\EndIf
\EndWhile
		\State $\Upsilon^l\gets m$ (or $\Upsilon^s\gets m$)
\EndProcedure
\end{algorithmic}
\end{algorithm}

\section{Impact of Queueing on the Revenue Maximization Problem}
\label{sec:queuing_effect}
What is the impact of the queuing effect on the pricing in the revenue maximization problem (\ref{eq:pricing})? Let's consider the following instantaneous revenue maximization problem without the queueing shift.
\begin{equation}
\label{eq:revenue}
    \text{Maximize}\;\; qD(q,M).
\end{equation}
For simplicity, we ignore the time index in the discussion.

Note that both problems in \eqref{eq:pricing} and \eqref{eq:revenue} may have multiple optimal solutions. For the purpose of obtaining intuitions, we will restrict our discussion to the case where there is a unique optimal price for both \eqref{eq:pricing} and \eqref{eq:revenue}. To guarantee this, we assume that revenue $R(D)\mdef \hat{q}(D)D$ is a strictly concave function of the demand\footnote{This assumption is common in the revenue management literature (\eg \cite{talluri2005theory}) to guarantee unique optimal pricing.}, where $\hat{q}(D)$ is defined as the inverse demand function, \ie $\hat{q}(D)\mdef \max\{\hat{q}: D(\hat{q},m)=D\}$\footnote{Since the demand function $D(q,m)$ is non-increasing in $q$, there can be multiple prices resulting the same demand.} for a given $m$. For simplicity, we denote the optimal price in revenue maximization \eqref{eq:pricing} as $q^*$, and the optimal price in  revenue maximization problem \eqref{eq:revenue} with the queuing shift as $q^{**}$. We will show that $q^{**}\ge q^*$, \ie the queuing effect leads to a higher price.

The objective of revenue maximization problem in (\ref{eq:revenue}) can be represented as $R(D)= \hat{q}(D)D$, and its optimal demand is denoted as $D^*$. By the first order optimality condition, $D^*$ satisfies that $R'(D^*)=0$, where $R'(\cdot)$ denotes the first order derivative of $R(\cdot)$.
When the queuing effect is taken into consideration, the objective of (\ref{eq:pricing}) can be represented as $R(D)-\frac{D Q}{V}$, and we denote its optimal demand as $D^{**}$. Again by the first optimality condition, $D^{**}$ satisfies that $R'(D^{**})=\frac{Q}{V}\ge 0$. Since the revenue function $R(D)$ is concave in $D$, $R'(D)$ is a decreasing function. Since $R'(D^{**})\ge R' (D^*)$, we obtain $D^{**}\le D^*$.
Furthermore, since the demand $D(q,m)$ is non-increasing in price $q$ for a given $m$, we have $q^{**}\ge q^*$. In other word, when incorporating the queuing effect, the optimal dynamic price $q^{**}$ in the PMC policy is higher than the optimal price $q^*$ in instantaneous revenue maximization problem without the shift. Moreover, the larger the queue length $Q$,  the higher the dynamic price $q$ in the PMC policy. When we perform such pricing in the system, a high price will decrease the demand, which will slow the increase of the queue length. Thus the dynamic pricing in the PMC policy also performs the functionality of congestion control to some extent.

For a more general case where the concavity assumption may not be satisfied, the queueing effect depends on the shape of the revenue function at the point $D^*$, \ie $q^{**}\ge q^{*}$ if and only if $R'(D^*)\le \frac{Q}{V}$.

\end{document}